\documentclass[runningheads,envcountsect]{llncs}

\usepackage[utf8]{inputenc}
\usepackage{amsmath,amssymb}
\usepackage{proof}
\usepackage{mathtools}
\usepackage{xfrac}
\usepackage{physics}
\usepackage{tensor}
\usepackage{stmaryrd}
\usepackage{thm-restate}
\usepackage{stmaryrd}
\usepackage{wasysym}
\usepackage{booktabs}
\usepackage{tikz}
\usetikzlibrary{arrows.meta,calc,positioning,decorations.pathmorphing,fit,shapes.misc,patterns,patterns.meta}
\usepackage{subcaption}
\usepackage{apxproof}
\usepackage{hyperref}
\usepackage{soul}
\usepackage{extarrows}

\usepackage{orcidlink}  

\makeatletter
\newcommand*\mysize{%
  \@setfontsize\mysize{8.5}{9.0}%
}
\makeatother

\newcommand{\nil}{\mathbf{0}}
\newcommand{\sop}[3][{}]{\mathcal{#2}_{#1}\!\left(#3\right)\!}
\newcommand{\E}{\mathcal{E}}
\newcommand{\nE}{\dot\E}
\newcommand{\unitary}[3][{}]{\text{#2}_{#1}(#3)}
\newcommand{\meas}[3][M]{\mathbb{M}_{#1}(#2 \rhd{} #3)}
\newcommand{\tmeas}[2]{\mathbb{M}_{23}^{01}(#1 \rhd{} #2)}
\newcommand{\varset}{\text{Var}}
\newcommand{\opset}{\text{Op}}
\newcommand{\measset}{\text{Meas}}
\newcommand{\chanset}{\text{Chan}}
\newcommand{\chtype}[1]{\widehat{#1}}
\newcommand{\ntype}{\mathbb{N}}
\newcommand{\btype}{\mathbb{B}}
\newcommand{\qtype}{\mathcal{Q}}
\newcommand{\confset}{\mathit{Conf}}

\newcommand{\tagset}{\mathit{Tag}}
\newcommand{\choiceset}{\mathit{Sched}}
\newcommand{\actset}{\mathit{Act}}

\newcommand{\hilbert}{\mathcal{H}}
\newcommand{\hilb}{\hilbert}
\newcommand{\qubit}{\widehat{\mathcal{H}}}
\newcommand{\qubits}[1]{\widehat{\mathcal{H}}^{\otimes{}#1}}
\newcommand{\ite}[3]{\textbf{if } #1 \textbf{ then } #2 \textbf{ else } #3}
\newcommand{\proc}[1]{\text{\textbf{#1}}}
\newcommand{\iconf}{\mathcal{C}}

\newcommand{\confl}{\mathmbox{\big\langle\hspace{-4.4pt}\big\langle}}
\newcommand{\confr}{\mathmbox{\big\rangle\hspace{-4.4pt}\big\rangle}}
\newcommand{\lrconf}[1]{\left\langle{#1}\right\rangle}
\newcommand{\slrconf}[1]{\overline{\left\langle{#1}\right\rangle}}
\newcommand{\conf}[1]{\confl{#1}\confr}
\newcommand{\confel}{\mathcal{C}}

\newcommand{\singleton}[1]{\overline{#1}}
\newcommand{\distelem}[2]{#1 \bullet #2}
\newcommand{\psum}[1]{\tensor[_{#1}]{\oplus}{}}
\newcommand{\rel}{\mathcal{R}}
\newcommand{\dist}[1]{\mathcal{D}(#1)}
\newcommand{\sdist}[1]{\mathcal{D}^{\leq}(#1)}
\newcommand{\density}[1]{\mathcal{DO}(#1)}

\newcommand{\sdensity}[1]{\mathcal{DO}^{\leq}(#1)}

\newcommand{\soset}[1]{\mathcal{SO}^{\leq}(#1)}
\newcommand{\tsoset}[1]{\mathcal{SO}(#1)}
\newcommand{\cnot}{\mathit{CNOT}}
\newcommand{\I}{\mathbb{I}}
\newcommand{\cfield}{\mathbb{C}}
\newcommand{\mass}[1]{|#1|}
\newcommand{\support}[1]{\lceil#1\rceil}

\DeclareMathOperator{\lift}{lift}

\newcommand{\tauarrow}{\xlongrightarrow{\tau}}

\newcommand{\oost}{\frac{1}{\sqrt{2}}}
\newcommand{\kp}{\ket{\psi}}
\newcommand{\kf}{\ket{\phi}}
\newcommand{\kz}{\ket{0}}
\newcommand{\ko}{\ket{1}}
\newcommand{\kpl}{\ket{+}}
\newcommand{\km}{\ket{-}}
\newcommand{\dket}[1]{\ketbra{#1}}

\newcommand{\confbot}{\confset_{\!\!\bot}}
\newcommand{\sconf}[1]{\singleton{\conf{#1}}}
\newcommand{\true}{\mathit{tt}}
\newcommand{\false}{\mathit{ff}}
\newcommand{\tags}{\colon\!}

\newcommand{\rulename}[1]{\mbox{\scriptsize\scshape #1}}
\newcommand{\blank}{{\,\cdot\,}}

\newcommand{\com}[1]{{\color{red} #1}}

\newtheorem{assumption}{Assumption}

\usetikzlibrary{decorations.pathmorphing,positioning,calc}
\tikzset{arrow/.style={-stealth}}
\tikzset{tarrow/.style={-{Straight Barb}{Straight Barb}}}
\tikzset{parrow/.style={-stealth,decoration={snake,amplitude=1pt,segment length=6pt,post length=2pt},decorate,thick}}
\tikzset{conf/.style={font=\mysize}}
\tikzset{ci/.style={font=\mysize,draw, circle}}
\tikzset{psplit/.style={draw, circle, fill=black, inner sep=1.5pt}}
\tikzset{hsplit/.style={inner sep=0pt}}
\tikzset{csplit/.style={fill, circle, inner sep=2pt}}
\tikzset{weight/.style={font=\mysize,midway}}

\newtheorem{fact}{Fact}[section]

\begin{document}

\title{Quantum Bisimilarity is a Congruence\\ under Physically Admissible Schedulers\thanks{Research carried out within the National Centre on HPC, Big Data and Quantum Computing - SPOKE 10 (Quantum Computing) and partly funded from the European Union Next-GenerationEU - National Recovery and Resilience Plan (NRRP) – MISSION 4 COMPONENT 2, INVESTMENT N. 1.4 – CUP N. I53C22000690001.}
}

\author{Lorenzo Ceragioli\inst{1}
	\orcidlink{0000-0002-1288-9623} 
	\and
	Fabio Gadducci\inst{2}
	\orcidlink{0000-0003-0690-3051} 
	\and
	\newline Giuseppe Lomurno\inst{2}
	\orcidlink{0009-0000-0573-7974} 
	\and
	Gabriele Tedeschi\inst{2}
	\orcidlink{0009-0002-5345-9141}}
%
\authorrunning{L. Ceragioli et al.}
\institute{IMT School for Advanced Studies, Lucca, Italy \\
	\email{lorenzo.ceragioli@imtlucca.it}
	\and
	University of Pisa, Pisa, Italy \\
	\email{fabio.gadducci@unipi.it}
	\email{\{giuseppe.lomurno,gabriele.tedeschi\}@phd.unipi.it}
}

\titlerunning{Quantum Bisimilarity is a Congruence}

\maketitle

\begin{abstract}
	The development of quantum 
	algorithms and protocols calls for adequate modelling and verification techniques, which
	requires abstracting and focusing on the basic features of quantum concurrent systems, like CCS and CSP have done for their classical
	counterparts.
	So far, an equivalence relation is still missing that
	is a congruence for parallel composition and adheres to the limited discriminating power implied by quantum theory.
	In fact, defining an adequate bisimilarity for quantum-capable, concurrent systems proved a difficult task, because unconstrained non-determinism allows to spuriously discriminate indistinguishable quantum systems. 
	We investigate this problem by enriching a linear quantum extension of CCS with simple physically admissible schedulers.
	We show that our approach suffices for deriving a well-behaved bisimilarity that satisfies the aforementioned desiderata.
\keywords{Quantum process calculi  \and bisimulation congruences.}
\end{abstract}

\section{Introduction}

%

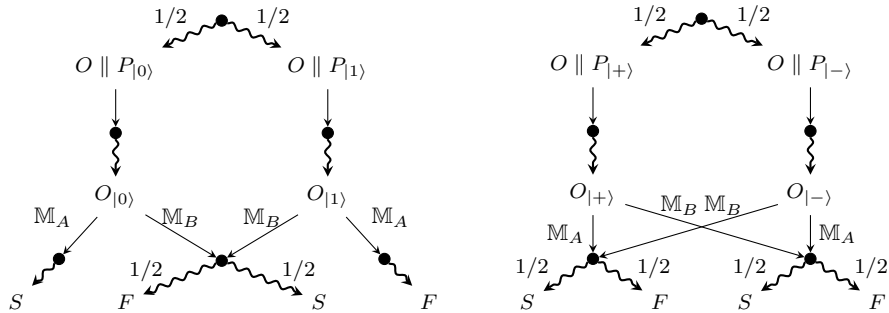
\begin{figure}[t]
	\centering
	\begin{tikzpicture}[node distance=5mm]
		\node [psplit] at (0,0) (p1) {};
		\node [conf, below left=3mm and 7mm of p1] (c12) {$O \parallel P_{\kz}$};
		\node [conf, below right=3mm and 7mm of p1] (c13) {$O \parallel P_{\ko}$};
		\node [psplit, below=of c12] (p12) {};
		\node [psplit, below=of c13] (p13) {};
		\node [conf, below=of p12] (c2) {$O_{\kz}$};
		\node [conf, below=of p13] (c3) {$O_{\ko}$};
		\node [psplit, below left=5mm and 3mm of c2] (p2) {};
		\node [psplit, below=of $(c2.south)!0.5!(c3.south)$] (p3) {};
		\node [psplit, below right=5mm and 3mm of c3] (p4) {};
		\node [conf, below left=3mm and 3mm of p2] (c4) {$S$};
		\node [conf, below right=3mm and 6mm of p2] (c5) {$F$};
		\node [conf, below left=3mm and 6mm of p4] (c6) {$S$};
		\node [conf, below right=3mm and 3mm of p4] (c7) {$F$};

		\draw [parrow] (p1) -- (c12) node[weight,above left=-.3mm and -1mm] {$1/2$};
		\draw [parrow] (p1) -- (c13) node[weight,above right=-.3mm and -1mm] {$1/2$};
		\draw [arrow] (c12) -- (p12) node[weight,left=1mm] { };
		\draw [arrow] (c13) -- (p13) node[weight,left=1mm] { };
		\draw [parrow] (p12) -- (c2) node[weight,above left=.3mm and 1mm] { };
		\draw [parrow] (p13) -- (c3) node[weight,above right=.3mm and 1mm] { };
		\draw [arrow] (c2) -- (p2) node[weight,above left] { $\mathbb{M}_A$ };
		\draw [arrow] (c2) -- (p3) node[weight,above] { $\mathbb{M}_B$ };
		\draw [arrow] (c3) -- (p3) node[weight,above] { $\mathbb{M}_B$ };
		\draw [arrow] (c3) -- (p4) node[weight,above right] { $\mathbb{M}_A$ };
		\draw [parrow] (p2) -- (c4) node[weight,above left=-1.5mm and 1mm] { };
		\draw [parrow] (p3) -- (c5) node[weight,above left=-1.5mm and 1mm] {$1/2$};
		\draw [parrow] (p3) -- (c6) node[weight,above right=-1.5mm and 1mm] {$1/2$};
		\draw [parrow] (p4) -- (c7) node[weight,above right=-1.5mm and 1mm] { };
	\end{tikzpicture}
	\hspace{.6cm}
	\begin{tikzpicture}[node distance=5mm]
		\node [psplit] at (0,0) (p1) {};
		\node [conf, below left=3mm and 7mm of p1] (c12) {$O \parallel P_{\kpl}$};
		\node [conf, below right=3mm and 7mm of p1] (c13) {$O \parallel P_{\km}$};
		\node [psplit, below=of c12] (p12) {};
		\node [psplit, below=of c13] (p13) {};
		\node [conf, below=of p12] (c2) {$O_{\kpl}$};
		\node [conf, below=of p13] (c3) {$O_{\km}$};
		\node [psplit, below=of c2] (p2) {};
		\node [psplit, below=of c3] (p3) {};
		\node [conf, below left=3mm and 6mm of p2] (c4) {$S$};
		\node [conf, below right=3mm and 6mm of p2] (c5) {$F$};
		\node [conf, below left=3mm and 6mm of p3] (c6) {$S$};
		\node [conf, below right=3mm and 6mm of p3] (c7) {$F$};

		\draw [parrow] (p1) -- (c12) node[weight,above left=-.3mm and -1mm] {$1/2$};
		\draw [parrow] (p1) -- (c13) node[weight,above right=-.3mm and -1mm] {$1/2$};
		\draw [arrow] (c12) -- (p12) node[weight,left=1mm] { };
		\draw [arrow] (c13) -- (p13) node[weight,left=1mm] { };
		\draw [parrow] (p12) -- (c2) node[weight,above left=.3mm and 1mm] { };
		\draw [parrow] (p13) -- (c3) node[weight,above right=.3mm and 1mm] { };
		\draw [arrow] (c2) -- (p2) node[weight,left] { $\mathbb{M}_A$ };
		\draw [arrow] (c2) -- (p3) node[weight,above left=1mm and 1mm] { $\mathbb{M}_B$ };
		\draw [arrow] (c3) -- (p2) node[weight,above right=1mm and 1mm] { $\mathbb{M}_B$ };
		\draw [arrow] (c3) -- (p3) node[weight,right] { $\mathbb{M}_A$ };
		\draw [parrow] (p2) -- (c4) node[weight,above left=-1.5mm and 1mm] {$1/2$};
		\draw [parrow] (p2) -- (c5) node[weight,above right=-1.5mm and 1mm] {$1/2$};
		\draw [parrow] (p3) -- (c6) node[weight,above left=-1.5mm and 1mm] {$1/2$};
		\draw [parrow] (p3) -- (c7) node[weight,above right=-1.5mm and 1mm] {$1/2$};
	\end{tikzpicture}

	\caption{Observer in parallel with indistinguishable qubit sources.}
	\label{fig:ex-zopm}
\end{figure}

Recent years have seen a flourishing development of \emph{quantum computation} and \emph{quantum communication} technologies.
Both of them exploit quantum phenomena like superposition and entanglement to achieve quantitative advantages with respect to their classical counterparts.
The former focuses on the (supposedly) higher computational power of quantum computers, while the latter on security and reliability properties of communication, featuring solutions for key distribution~\cite{nurhadiQuantumKeyDistribution2018}, cryptographic coin tossing~\cite{bennettQuantumCryptographyPublic2014}, direct communication~\cite{longQuantumSecureDirect2007}, and private information retrieval~\cite{gaoQuantumPrivateQuery2019}.
Quantum communication also promises to
allow
linking multiple computers via the \emph{Quantum Internet}~\cite{caleffiQuantumInternetCommunication2018,zhangFutureQuantumCommunications2022},
therefore providing
quantum algorithms with large enough memories for practical applications.

With these advances, the need emerged for modelling and verification techniques
applicable to quantum distributed algorithms and protocols,
but an accepted standard is still missing.
Numerous works~\cite{lalireprocess2004,gaycommunicating2005,fengbisimulation2012,ceragioliQuantumBisimilarityBarbs2024}
rely on \textit{quantum process calculi}, an algebraic formalism successfully applied to classical protocols and concurrent systems.
While the features of these calculi are mostly comparable, the bisimilarities greatly vary.
The desiderata for a bisimilarity relation are that it is a congruence for parallel composition and that it adheres to the limited discriminating power implied by quantum theory.
	Furthermore, it should make the atomic observable properties explicit through a labelled approach, thus making the verification simpler and more efficient to implement.
However, none of the previous proposals yields a relation satisfying all these criteria.

Concerning the discriminating power, quantum theory prescribes that the state of a quantum system cannot be observed directly, but only through \emph{measurements}, which have a probabilistic outcome and cannot avoid altering the state they are measuring.
Moreover, there are different kinds of measurements, each capable of discriminating only some of the different states, while equating others.
This constraint limits the capability of 
discerning
the behaviour of quantum systems.
Even though quantum process calculi only allow measurements to inspect the quantum state,~\cite{kubotaapplicationnoyear,davidsonformal2012} showed that some of the proposed bisimilarities behave as if they can implicitly compare quantum values, and highlighted some non-bisimilar processes that should be indistinguishable.
In~\cite{ceragioliQuantumBisimilarityBarbs2024}, the authors prove that this discrepancy with respect to quantum theory is due to the discriminating power of non-deterministic choices.
If not constrained to be based on classical information, non-determinism allows processes to act according to unknown quantum values, implicitly revealing them.

Consider e.g. a family of processes $P_{\kz}, P_{\ko}, P_{\kpl}$, and $P_{\km}$, with $P_{\kp}$ sending a qubit in state $\kp$.
Take two qubit sources: the first sends a qubit either in state $\kz$ or $\ko$ with equal probability, the second sends a qubit either in $\kpl$ or $\km$.
Quantum theory deems the values of the received qubits indistinguishable, as they yield the same result under any possible measurement.
Nonetheless, consider an observer $O$ that receives the qubit and chooses non-deterministically which measurement to perform between
$\mathbb{M}_A$ and $\mathbb{M}_B$: the former telling apart $\kz$ from $\ko$ and equating $\kpl$ and $\km$, the latter equating $\kz$, $\ko$, $\kpl$ and $\km$.
\autoref{fig:ex-zopm} presents the evolution of the two sources
paired with the observer, where $O_{\kp}$ is the observer after receiving a qubit in state $\kp$, and the immediately distinguishable states $S$ and $F$ are chosen according to the result of the measurement.
Straight arrows models actions, while the squiggly ones represent the elements of a distribution, labelled by the probability (we omit $1$).
Take the system on the left.
If the measurement is chosen according to the value of the received qubit, the observer can perform $\mathbb{M}_A$ when receiving $\kz$ and $\mathbb{M}_B$ otherwise, therefore obtaining the state $S$ with probability $3/4$.
The system on the right cannot replicate this behaviour, thus $O$ distinguishes the two sources.
This contradicts the prescriptions of quantum theory.
However, it should be impossible to know the quantum state without first performing a measurement, hence such a combination of non-deterministic choices is not physically plausible.

In this paper,
we introduce a labelled version of the operational semantics of lqCCS~\cite{ceragioliQuantumBisimilarityBarbs2024}, a linearly typed quantum extension of CCS,
and we investigate
schedulers and scheduled bisimilarity~\cite{chatzikokolakisBisimulationDemonicSchedulers2009}.
Schedulers are usually employed to characterize ``admissible'' or ``realistic'' choices, e.g. the ones agnostic with respect to private data.
We use them to make non-deterministic choices of lqCCS compliant with quantum theory.

Our main result is that simple syntactic schedulers suffice for recovering a notion of bisimilarity that satisfies all our desiderata.
More in detail, we mark the available choices of processes with tags that do not
depend on quantum values, and we force schedulers to choose based on tags only.
We define a scheduled version of (probabilistic) saturated bisimilarity $\sim_s$,
pairing processes that are indistinguishable under any context and scheduler.
Our proposed $\sim_s$ is a congruence for parallel composition, and satisfies the indistinguishability property introduced by~\cite{ceragioliQuantumBisimilarityBarbs2024}, which lifts the known equivalence of quantum values to lqCCS processes.
Finally, we derive a labelled bisimilarity $\sim_l$, and we prove
it equivalent to $\sim_s$.
This characterizes the atomic observable properties of quantum capable processes, and paves the way for automated verification.

\paragraph{\textbf{Related Works.}}
Our work proposes a saturated and a labelled bisimilarity, both based on schedulers, and proves them equivalent.
The choice of saturated bisimilarity as the milestone behavioural equivalence is introduced in~\cite{ceragioliQuantumBisimilarityBarbs2024}, where it is shown that unconstrained non-deterministic contexts do not comply with the observational limitations of quantum theory.
That paper presents a different semantics for processes and contexts, constraining non-determinism only of the latter.
Here we instead treat them uniformly by adding tags to and constraining non-determinism in both of them, obtaining a congruence with respect to parallel composition.
Moreover, we additionally propose an equivalent labelled semantics,
explicitly representing the observable properties of concurrent quantum systems.

Tagged
processes have been introduced in~\cite{chatzikokolakisMakingRandomChoices2007}, and are vastly used in
probabilistic systems for characterizing which choices are ``admissible'' or ``realistic''~\cite{canettiTimeBoundedTaskPIOAsFramework2006,andresInformationHidingProbabilistic2011,songDecentralizedBisimulationMultiagent2015}.
Our tags are reminiscent of the ones used by~\cite{chatzikokolakisMakingRandomChoices2007,chatzikokolakisBisimulationDemonicSchedulers2009} to prevent schedulers from choosing a move based on private data.
In the same works, the authors show that different tagging policies correspond to different classes of schedulers.
Our usage of tags is somehow similar, as we use them to impose limitations on what can affect the choice of schedulers, but our constraints are motivated by
the physical limitations prescribed by quantum theory instead of secrecy assumptions.
The schedulers of~\cite{chatzikokolakisMakingRandomChoices2007,chatzikokolakisBisimulationDemonicSchedulers2009} are \emph{deterministic},
meaning that they choose a transition in a deterministic manner.
We extend them to randomized schedulers, which may choose how to reduce probabilistically~\cite{Segala95}.
However, our schedulers lack conditionals, which are present in~\cite{chatzikokolakisMakingRandomChoices2007}.
This is intended as we focus on simple, constrained schedulers and not general ones.
We use schedulers as subscripts of the transition relation, as is Section 5 of~\cite{chatzikokolakisBisimulationDemonicSchedulers2009}.

An equivalence between saturated and labelled bisimilarity for quantum protocols was introduced in~\cite{dengopen2012}, but their bisimilarity was later found too strict to adhere to the prescriptions of quantum theory~\cite{kubotasemi-automated2016,fengtoward2015-1}.
The most recent labelled bisimilarity for their calculus has been presented in~\cite{dengbisimulations2018}, which proposes a semantics made of sub-distributions, and compares the visible qubits of sub-distributions as a whole, instead of comparing single configurations. 
However, no general property is proved about the adherence with quantum theory, and the proposed bisimilarity is not a
congruence with respect to the parallel operator, as shown in~\cite{ceragioliQuantumBisimilarityBarbs2024}.
We employ the same sub-distribution approach of~\cite{dengbisimulations2018} for our labelled bisimilarity, 
making it a congruence thanks to the adequate treatment of constrained schedulers, 
and we prove that our version correctly relates processes acting on indistinguishable quantum states.


\paragraph{\textbf{Synopsis.}}
In Section~\ref{sec:bg} we give some background about
probability distributions and quantum computing.
In Section~\ref{sec:lqCCS} we present our scheduled semantics for
lqCCS.
In Section~\ref{sec:db} we
propose our behavioural equivalence and we investigate its properties, also providing an equivalent labelled characterization.
Finally,
we wrap-up the paper in Section~\ref{sec:conc}.
The full proofs and all the details about the type systems are postponed to the appendices.

\section{Background}\label{sec:bg}
We recall some background on probability distributions
and quantum computing, referring to~\cite{nielsenquantum2010} for further reading.

\subsection{Probability Distributions}\label{probabilisticBackground}
A \emph{sub-probability distribution} over a set $S$ is a function $\Delta\colon\! S \to
	[0,1]$ such that $\sum_{s \in S}\Delta(s) \leq 1$.
When a distribution $\Delta$ satisfies $\sum_{s \in S}\Delta(s) = 1$ we say it is a \emph{probability distribution}.
We call the \emph{support} of a distribution $\Delta$ the set $ \support{\Delta} = \{ s \in
	S\;|\;\Delta(s) > 0 \}$.
We denote with $\sdist{S}$ the set of sub-probability distributions with finite support, and with $\dist{S} \subsetneq \sdist{S}$ the probability ones.


For each $s \in S$, we let $\singleton{s}$ be the \emph{point distribution} $\singleton{s}(s) = 1$.
Given a finite set of non-negatives reals $\{p_i\}_{i \in I}$ such that $\sum_{i \in I} p_i = 1$, we write $\sum_{i \in I} \distelem{p_i}{\Delta_i}$ for the distribution
determined by $(\sum_{i \in I} \distelem{p_i}{\Delta_i})(s) = \sum_{i \in
		I}p_i\Delta_i(s)$.
The notation $\Delta_1 \psum{p} \Delta_2$ is a shorthand for
$\distelem{p}{\Delta_1} + \distelem{(1 - p)}{\Delta_2}$.

A relation $\mathcal{R} \subseteq \dist{S} \times \dist{S}$ is said to be
\emph{linear} if $(\Delta_1 \psum{p} \Delta_2)\;\mathcal{R}\;(\Theta_1 \psum{p}
	\Theta_2)$ for any $p \in [0, 1]$ whenever $\Delta_i\;\mathcal{R}\;\Theta_i$
for $i = 1, 2$.
Moreover, $\rel$ is \emph{left-decomposable} if $(\Delta_1 \psum{p}
	\Delta_2)\;\mathcal{R}\; \Theta$ implies $\Theta = (\Theta_1 \psum{p}
	\Theta_2)$ for some $\Theta_1, \Theta_2$ with $\Delta_i\;\mathcal{R}\;\Theta_i$ for $i = 1,2$ and for any $p \in [0, 1]$. Right-decomposability is defined symmetrically, and a relation is \emph{decomposable} when it is both left- and right-decomposable.

Given $\mathcal{R} \subseteq A \times \dist{B}$, its \emph{lifting}
$\text{lift}(\mathcal{R}) \subseteq \dist{A} \times \dist{B}$ is the smallest
linear relation such that $\overline{s}\;\text{lift}(\mathcal{R})\;\Theta$ when
$s\;\mathcal{R}\;\Theta$.
We generalize this
to
$\lift_{A_i}(\rel) \subseteq (A_1 \times \dots \times \dist{A_i} \times \dots \times A_n) \times \dist{B}$
for the lifting on the $i$-th component of an $n$-ary relation $\rel \subseteq (A_1 \times \dots \times A_i \times \dots \times A_n) \times \dist{B}$, obtained by
fixing the other arguments.
%
For example, $\lift_{A_n} (\rel)$ is defined as
%
\[
	\{ (a_1, \dots, a_{n-1}, \Delta, \Theta) \mid (\Delta, \Theta) \in \lift(\{ (a_n, \Theta) | (a_1, ... a_{n-1}, a_n, \Theta) \in R \}) \}.
\]
%
%
%
%
\subsection{Quantum Computing}

An isolated physical system is associated to a \emph{Hilbert space} $\hilb$,
i.e. a complex vector space equipped with an inner product $\braket\blank$.
We indicate column vectors as $\kp$ and their conjugate transpose as
$\bra\psi = \kp^\dagger$.
The states of a system are \emph{unit vectors} in $\hilb$, i.e. vectors $\kp$ such that $\braket\psi =
	1$.
A two-dimensional physical system is known as a \emph{qubit}, and we denote its Hilbert space as $\qubit = \mathbb{C}^2$.
The vectors
$\ket0 = (1,0)^T$ and $\ket1 = (0,1)^T$
form the \emph{computational basis} of
$\qubit$. Other important states are $\ket+ = \frac{1}{\sqrt{2}}(\ket0 + \ket1)$ and $\ket- = \frac{1}{\sqrt{2}}(\ket0 - \ket1)$, which form the \emph{Hadamard basis}.
In the quantum jargon, the states in the Hadamard basis are \emph{superpositions} with respect to the computational basis, as they are a linear combination of $\kz$ and $\ko$.
A third basis of $\qubit$ contains $\ket{i} = \oost(\kz + i\ko)$ and $\ket{-i} = \oost(\kz -i\ko)$.
Both $\kpl, \km$ and $\ket{i}, \ket{-i}$ are uniform superpositions of $\kz, \ko$, but they differ in the \emph{phase} of the $\ko$ coefficient.

%

We represent the state space of a composite physical system as the \emph{tensor
	product} of the state spaces of its components.
Consider the Hilbert spaces $\hilb_A$ with $\{\ket{\psi_i}\}_{i \in I}$ one of its bases,
and $\hilb_B$ with $\{\ket{\phi_j}\}_{j \in I}$ one of its bases.
We let their  tensor product $\hilb_A \otimes \hilb_B$ be the Hilbert space with bases $\{\ket{\psi_i} \otimes \ket{\phi_j}\}_{(i,j) \in I \times J}$,
where $\kp \otimes \kf$ is the Kronecker product.
We often omit the tensor product and write $\ket{\psi\phi}$ for $\ket{\psi}\otimes\ket{\phi}$.
We write $\qubits{n}$ for the $2^n$-dimensional Hilbert space defined as the tensor product of $n$
copies of $\qubit$ (i.e.\ the possible states of $n$ qubits).
A quantum state in $\hilbert_A \otimes \hilbert_B$ is  \emph{separable}
when it can be expressed as the Kronecker product of two vectors of
$\hilbert_A$ and $\hilbert_B$. Otherwise, it is \emph{entangled},
like the Bell state $\ket{\Phi^+} = \frac{1}{\sqrt{2}}(\ket{00} +
	\ket{11})$.

In quantum physics, the evolution of an
isolated system is described by a unitary transformation.
For each linear operator $A$ on $\hilbert$,
its \emph{adjoint} $A^\dag$
is the unique linear operator such that
$\mel{\psi}{A}{\phi} = \braket{A^\dag\psi}{\phi}$.
A linear operator $U$ is \emph{unitary} when $UU^\dag = U^\dag U = \I$, with $\I$ the identity matrix.
Quantum computers allow the programmer to manipulate registers via unitaries like
$H$, $X$, $Z$ and $\cnot$, satisfying
$H \kz = \kpl$ and $H \ko = \km$;
$X \kz = \ko$ and $X \ko = \kz$; $Z \kpl = \km$ and $Z \km = \kpl$;
$\cnot \ket{10} = \ket{11}$, $\cnot \ket{11} = \ket{10}$ and
$\cnot\ket{0\psi} = \ket{0\psi}$ (all the other cases are defined by linearity).


\subsection{Density operator formalism}

The density operator formalism puts together quantum systems and probability distributions by considering
mixed states, i.e.\ \emph{sub-probability distributions of quantum states}.
A point distribution $\singleton{\kp}$ (called a pure state) is represented by the matrix $\ketbra{\psi}$.
In general, a mixed state $\Delta \in \sdist{\qubits{n}}$ for $n$ qubits is represented as the matrix $\rho_\Delta \in \cfield^{2^n\times 2^n}$, known as its \emph{density operator}, with $\rho_\Delta = \sum_{\kp} \Delta(\kp)
		\dket{\psi}$.
We write $\sdensity{\hilbert}$ for the set of density operators of $\hilbert$ and $\density{\hilbert}$ for the set $\{\rho \in \sdensity{\hilb} \mid tr(\rho) = 1\}$, corresponding to probability distributions of quantum states.
For example, the mixed state $\singleton{\kz} \psum{1/3} \singleton{\kpl}$ is represented as the density operator
$		\frac{1}{3} \ketbra{0} + \frac{2}{3} \ketbra{+}
$.

Note that the encoding of probabilistic mixtures of quantum states as density operators is not injective.
For example, $\frac{1}{2}\I$ is called the \emph{maximally mixed state} and
represents both distributions $\Delta_{C} = \singleton{\ket0} \psum{1/2}
	\singleton{\ket1}$ and $\Delta_{H} = \singleton{\ket+} \psum{1/2}
	\singleton{\ket-}$.
This is a desired feature, as the laws of quantum mechanics deem indistinguishable all the distributions that result in the same density operator.
\begin{fact}\label{thm:qind}
	Two distributions of pure quantum states $\Delta, \Theta \in \dist{\hilb}$
	are indistinguishable for any physical observer whenever $$\rho_\Delta = \sum\nolimits_{\kp} \Delta(\kp) \dket{\psi} = \sum\nolimits_{\kp} \Theta(\kp) \dket{\psi} = \rho_\Theta.$$
\end{fact}

Since density operators are just distributions of pure states, the same result is easily extended to indistinguishable distributions of mixed states.

When modelling composite systems, density operators are composed with the Kronecker product as well.
Differently from pure states, they can also describe local information about subsystems.
Let $\hilbert_{AB} = \hilbert_A \otimes \hilbert_B$
represent a composite system, with subsystems $A$ and $B$.
Given a (not
necessarily separable) $\rho^{AB} \in \hilbert_{AB}$, the state of
the subsystem $A$ is described as the \emph{reduced density
	operator} $\rho^A = \tr_B(\rho^{AB})$, with $\tr_B$ the \emph{partial trace over $B$}, defined as
the linear transformation such that
$\tr_B(\ketbra{\psi}{\psi'} \otimes \ketbra{\phi}{\phi'}) =
	\ketbra{\psi}{\psi'}\tr(\ketbra{\phi}{\phi'})$.

When applied to pure separable states, the partial trace returns the actual
state of the subsystem. When applied to an entangled state, instead, it
produces a mixed state, because ``forgetting'' the
information on the subsystem $B$ leaves us with only partial
information on subsystem $A$. For example, the partial trace over the first
qubit of $\dket{\Phi^+}$ is the maximally mixed state.

The dynamics of mixed states is given by \emph{trace non-increasing
	superoperators}, i.e. functions on density operators.
A superoperator $\mathcal{E}: \sdensity{\hilb} \to \sdensity{\hilb}$ on a $d$-dimensional Hilbert space $\hilb$ is a function defined by its \emph{Kraus operator sum decomposition} $\{E_i\}_{i = 1, \ldots, d^2}$, satisfying that
$\sop{E}{\rho} = \sum_i E_i\rho E_i^\dag$ and $\sum_i E_i^\dag E_i \sqsubseteq \I$,
where
$A \sqsubseteq B$ means that $B - A$ is a positive semidefinite matrix.
We call $\soset{\hilbert}$ the set of trace non-increasing superoperators on $\hilbert$, and
$\tsoset{\hilbert} \subseteq \soset{\hilbert}$ the set of all
\emph{trace-preserving} superoperators, i.e.\ such that $\sum_i E_i^\dag E_i = \I$.
Roughly, trace-preserving superoperators map distributions of quantum states to distributions of quantum states, while the result of applying a trace non-increasing superoperator to a distribution may be a \emph{sub}-probability distribution.

Noticeably, the tensor product of superoperators is obtained by tensoring their Kraus decompositions, and any
unitary transformation $U$ can be seen as a superoperator, that we still denote as $U$, having $\{ U \}$ as its Kraus
decomposition.

\emph{Quantum measurements} describe how to extract information from a physical system.
Performing a measurement on a quantum state returns a probabilistic classical result and causes the quantum state to change (i.e.\ to \emph{decay}).
%
%
A measurement with $k$ different outcomes is a set $\mathbb{M} = \{M_m\}_{m=0}^{k-1}$ of $k$ linear operators,
satisfying the \emph{completeness} equation $\sum_{m=0}^{k-1}
	M_m^\dag M_m = \I$.
If the state of the system is $\rho$ before the
measurement, then the probability of $m$-th outcome occurring is $p_m = tr(M_m \rho M_m^\dagger)$. If $m$ is the outcome, then the state after the measurement will be
$M_m \rho M_m^\dagger /p_m$.
Note that each operator $M_m$ defines a trace non-increasing superoperator $\mathcal{M}_m$, with $\mathcal{M}_m(\rho) = M_m \rho M_m^\dagger$, and the resulting state after the $m$-th outcome is the normalization of $\mathcal{M}_m(\rho)$.


The simplest measurements project a state into the elements of a basis,
e.g. $\mathbb{M}_{01} = \{\ketbra{0}, \ketbra{1}\}$ and $\mathbb{M}_\pm = \{\ketbra{+}, \ketbra{-}\}$ for the
computational and Hadamard basis of $\qubit$.
As expected, applying $\mathbb{M}_{01}$ to $\kz$ always returns the classical outcome $0$ and the state $\kz$.
When applying the same measurement on $\kpl$, instead, the result is $0$ and $\kz$, or $1$ and $\ko$ with equal probability.
Symmetrically, measuring $\kz$ with $\mathbb{M}_\pm$ leads to either $0$ and $\kpl$, or $1$ and $\km$, with equal probability.
Finally, applying the measurement $\mathbb{M}_{\pm i} = \{ \ketbra{i}, \ketbra{-i} \}$ to each of $\kz$, $\ko$, $\kpl$ and $\km$ returns $0$ and $\ket{i}$, or $1$ and $\ket{-i}$ with equal probability.

\section{A Quantum Process Algebra}\label{sec:lqCCS}
In the following sections we describe the syntax and the type system of lqCCS
processes, as well as a semantics decorated with schedulers.
Our process calculus is enriched with a linear type system,
reflecting the \emph{no-cloning theorem} of quantum mechanics, which forbids quantum values to be copied or broadcast.
Moreover, we introduce tags in the syntax, naming
all the possible non-deterministic choices,
and schedulers in the semantics, choosing among
the available tags.

\subsection{Syntax and Type System}\label{Syntax}

The syntax of tagged lqCCS processes is defined by the productions
\begin{align*}
	P \Coloneqq & \ t \tags \tau . P \mid (t,t) \tags \tau. P \mid t \tags \sop{E}{\tilde{e}}.P \mid t \tags \meas[]{\tilde{e}}{x}.P \mid t \tags c?x . P \mid \nil_{\tilde{e}} \\
	\mid        & \ t \tags c!e . P \mid \ite{e}{P}{P} \mid  P + P \mid P \parallel P \mid P \setminus c                                                                        \\
	e \Coloneqq & \ x \mid b \mid n \mid q \mid \neg e \mid e \lor e \mid e \leq e \mid e = e
\end{align*}
where $b \in \btype$, $n \in \ntype$, $q \in \qtype$, $x \in \varset$, $c \in
	\chanset$ with $\qtype$, $\varset$, $\chanset$ denumerable
sets of respectively qubit names, variables and channels, each typed.
The types for channels are $\widehat{\ntype}, \widehat{\btype}$ and $\widehat{\qtype}$.
We use $\tilde{e}$ to denote a (possibly empty) tuple $e_1, \ldots, e_n$ of expressions.

We assume a denumerable set $\tagset = \{t_1, t_2, \ldots\}$ of tags, and the actions are tagged with a tag $t \in \tagset$.
For the silent action $\tau$ we consider two possibilities, either a single tag or a pair of tags.
This is useful for using $\tau$ actions to write an abstract specification of a concrete protocol: $t \tags \tau$ models an action with tag $t$, $(t, t')\tags\tau$ models a synchronization.
The process $\nil_{\tilde{e}}$ \emph{discards} the qubits in ${\tilde{e}}$.
It behaves as a deadlock process that maintains ownership of the qubits in ${\tilde{e}}$ and makes them inaccessible to other processes.
Discard processes will be shown semantically equivalent to any deadlock process using the same qubits.
The process $\nil_{q_1,q_2}$ is, e.g., equivalent to $\nil_{q_2,q_1}$ and to $(c!q_1.c!q_2.\nil) \setminus c$, since $q_1$ and $q_2$ will never be available.
When $\tilde{e}$ is the empty sequence, we write $\nil$ to stress the equivalence with the nil process of standard CCS\@.
This feature of lqCCS allows marking which qubits are hidden to the environment, thus relieving bisimilar processes to agree on them.
A symbol $\mathcal{E}$ denotes a trace-preserving superoperator on
$\qubits{n}$ for some $n > 0$, and we write $\mathcal{E} : \opset(n)$ to
indicate that $\mathcal{E}$ is a superoperator of arity $n$.
A symbol $\mathbb{M}$ denotes a measurement $\{M_0,\ldots,M_{k-1}\}$ with $k$ different outcomes;
we write $\mathbb{M} : \measset(n)$ to indicate that each $M_i$ acts on $n$ qubits, and denote $|\mathbb{M}|$ the cardinality $k$ of $\mathbb{M}$.
Recall that $\mathbb{M}_{01}, \mathbb{M}_\pm$ and $\mathbb{M}_{\pm i}$ are the projective measurements in the bases $\{ \kz, \ko \}, \{ \kpl, \km \}$ and $\{ \ket{i}, \ket{-i} \}$.
Parallel composition, non-deterministic sum and restriction are the standard CCS ones.

As for lqCCS, the visibility of qubits is enforced explicitly through a linear type system.
The typing system of \cite{ceragioliQuantumBisimilarityBarbs2024} is applicable also to tagged lqCCS processes, since the tags are annotations that do not change the ownership of qubits.
The typing judgment $\Sigma \vdash P$ indicate that the process $P$ is well-typed under the usage of the set of qubits $\Sigma \subseteq \qtype$.
The typing for processes is unique~\cite{ceragioliQuantumBisimilarityBarbs2024},
i.e.
whenever $\Sigma \vdash P$ and $\Sigma' \vdash P$ then $\Sigma = \Sigma'$. For this reason we will call $\Sigma_P$ the only context which types $P$.
The type system is available in Appendix \ref{sec:typeappendix}.

\begin{example}\label{ex:quantumlottery}
	Consider a quantum lottery $\proc{QL}  = \proc{Pr} \parallel \proc{An} $ formed by processes $\proc{Pr}$, which prepares a qubit used as a source of randomness, and $\proc{An}$, which receives it, measures it, and  announces the winner between Alice and Bob.
	\begin{align*}
		\proc{Pr} & = (t_1\tags X(q).t_3\tags c!q.\nil) + (t_2\tags H(q).t_3\tags c!q.\nil)                              \\
		\proc{An} & = t_4\tags c?x.t_4\tags \meas[{01}]{x}{y}.\ite{y = 0}{t_5\tags a!1.\nil_{x}}{ t_6\tags b!1.\nil_{x}}
	\end{align*}
	Intuitively, $\proc{Pr}$ can prepare and send a qubit by applying either $X$ or $H$ to its qubit $q$, and $\proc{An}$ announces that Alice wins with $a!1$ if the qubit is found in state $\kz$, or that Bob wins with $b!1$ if the received qubit is found in state $\ko$.
	The unique typing of $\proc{QL}$ is $\{q\} \vdash
		\proc{QL}$, with $a, b : \chtype{\ntype}$, $c : \chtype{\qtype}$, $y : \ntype$,
	and $q, x : \qtype$.
\end{example}

\begin{figure}[!t]
	\footnotesize
	\begin{gather*}
		\infer[\rulename{ITEt}]{\conf{\rho, \ite{e}{P}{Q}} \xlongrightarrow{\mu}_{s} \Delta}{\conf{\rho,P} \xlongrightarrow{\mu}_s \Delta & {e \Downarrow \true}} \qquad
		\infer[\rulename{ITEf}]{\conf{\rho, \ite{e}{P}{Q}} \xlongrightarrow{\mu}_{s} \Delta}{\conf{\rho,Q} \xlongrightarrow{\mu}_s \Delta & {e \Downarrow \false}} \\[0.2cm]
		\infer[\rulename{Tau}]{\conf{\rho, t \tags \tau . P} \xlongrightarrow{\tau}_{t} \singleton{\conf{\rho, P}}}{} \qquad
		\infer[\rulename{TauPair}]{\conf{\rho, (t, t') \tags \tau . P} \xlongrightarrow{\tau}_{(t, t')} \singleton{\conf{\rho, P}}}{} \\[0.2cm]
		\infer[\rulename{Send}]{\conf{\rho, t \tags c!e . P} \xlongrightarrow{c!v}_{t} \singleton{\conf{\rho, P}}}{e \Downarrow v} \qquad
		\infer[\rulename{Receive}]{\conf{\rho, t \tags c?x . P} \xlongrightarrow{c?v}_{t} \singleton{\conf{\rho, P[\sfrac{v}{x}]}}}{c : \chtype{\qtype} \Rightarrow v \in \Sigma_\rho}\\[0.2cm]
		\infer[\rulename{Restrict}]{\conf{\rho, P \setminus c} \xlongrightarrow{\mu}_s \Delta \setminus c}{\conf{\rho, P} \xlongrightarrow{\mu}_s \Delta & {\mu \not \in \{c!v, c?v\}}} \qquad
		\infer[\rulename{QOp}]{\conf{\rho, t \tags \sop{E}{\tilde{q}} . P} \xlongrightarrow{\tau}_{t} \singleton{\conf{\mathcal{E}^{\tilde{q}}(\rho), P}}}{} \\[0.2cm]
		\infer[\rulename{QMeas}]{\conf{\rho, t\tags \meas[]{\tilde{q}}{y} . P} \xlongrightarrow{\tau}_{t} \sum_{m=0}^{|\mathbb{M}| - 1} \distelem{p_m}{\singleton{\conf{\frac{1}{p_m}\mathcal{M}_m^{\tilde{q}}(\rho), P[\sfrac{m}{y}]}}}}{p_m = \tr(\mathcal{M}_m^{\tilde{q}}(\rho))} \\[0.2cm]
		\infer[\rulename{ParL}]{\conf{\rho, P \parallel Q} \xlongrightarrow{\mu}_s \Delta \parallel Q}{\conf{\rho, P} \xlongrightarrow{\mu}_s \Delta & {\mu \not \in \{ c?v \mid  v \in \Sigma_Q \}}} \qquad
		\infer[\rulename{ParR}]{\conf{\rho, P \parallel Q} \xlongrightarrow{\mu}_s P \parallel \Delta}{\conf{\rho, Q} \xlongrightarrow{\mu}_s \Delta &{\mu \not \in \{ c?v \mid  v \in \Sigma_P \}}} \\[0.2cm]
		\infer[\rulename{SumL}]{\conf{\rho, P + Q} \xlongrightarrow{\mu}_{s} \Delta}{\conf{\rho,P} \xlongrightarrow{\mu}_s \Delta} \qquad
		\infer[\rulename{SynchL}]{\conf{\rho, P \parallel Q} \xlongrightarrow{\tau}_{(t,t')} \singleton{\conf{\rho, P' \parallel Q'}}}
		{ \conf{\rho, P} \xlongrightarrow{c!v}_t \singleton{\conf{\rho, P'}} &
			\conf{\rho, Q} \xlongrightarrow{c?v}_{t'} \singleton{\conf{\rho, Q'}}} \\[0.2cm]
		\infer[\rulename{SumR}]{\conf{\rho, P + Q} \xlongrightarrow{\mu}_{s} \Delta}{\conf{\rho,Q} \xlongrightarrow{\mu}_s \Delta} \qquad
		\infer[\rulename{SynchR}]{\conf{\rho, P \parallel Q} \xlongrightarrow{\tau}_{(t,t')} \singleton{\conf{\rho, P' \parallel Q'}}}
		{ \conf{\rho, Q} \xlongrightarrow{c!v}_t \singleton{\conf{\rho, Q'}} &
			\conf{\rho, P} \xlongrightarrow{c?v}_{t'} \singleton{\conf{\rho, P'}}}
	\end{gather*}

	\caption{Rules of lqCCS semantics.}
	\label{semantics}
\end{figure}

\subsection{Operational Semantics}\label{Semantics}

We describe a labelled semantics for lqCCS in terms of \emph{configurations} $\conf{\rho, P} \in \confset$, each composed by a global quantum state and a tagged lqCCS process.
Given a set $\Sigma = \{q_1,\ldots,q_n\} \subseteq \qtype$ and its associated Hilbert space $\hilbert_\Sigma = \qubits{n}$, a global
quantum state $\rho$ is a density operator in $\density{\hilbert_\Sigma}$.
The type system is extended to configurations by considering the qubits of the
underlying quantum state.
Let $\Sigma_\rho$ be the set of qubits appearing in $\rho$.
\begin{definition}
	Let $\conf{\rho, P} \in \confset$ and $\Delta \in \dist{\confset}$.
	We let
	$(\Sigma_\rho, \Sigma_P) \vdash \conf{\rho, P}$ if $\Sigma_P \subseteq \Sigma_\rho$, and $(\Sigma, \Sigma') \vdash \Delta$ if $(\Sigma, \Sigma') \vdash
		\iconf$ for any $\iconf \in \support{\Delta}$.
\end{definition}
Hereafter, we restrict ourselves to well-typed distributions,
and denote with $\Sigma_{\overline{P}}$ the set $\Sigma_\rho \setminus \Sigma_P$, i.e.\ the qubits only available to the environment.

Our semantics is decorated with \emph{syntactic schedulers}, which resolve non-determinism by choosing a tag.
The syntax of schedulers $s \in \choiceset$ for tagged lqCCS distributions is defined as follows
\[
	s \Coloneqq h \mid t \mid (t, t)
\]
A scheduler can stop the execution with the symbol $h$, select an action with a tag $t$, and select a synchronization with a pair $(t_1, t_2)$.
Intuitively, tags represent available visible options, upon which the scheduler can choose (possibly using a pair of them for synchronization).


Assume a set $Act$ of actions, containing $\tau$, $c!v$ and $c?v$ for any channel $c$ and value $v$.
The semantics of lqCCS is a probabilistic LTS (pLTS), i.e. a triple $(\confset, Act, \longrightarrow)$, with $\longrightarrow\ \subseteq \confset \times \actset \times \choiceset \times \dist{\confset}$.
A transition $(\conf{\rho, P}, \mu, s, \Delta) \in\ \longrightarrow$ is also denoted as $\conf{\rho, P} \xlongrightarrow{\mu}_s \Delta$.

The transition relation $\longrightarrow$ is the smallest relation
over configurations with closed processes
that satisfies the rules in~\autoref{semantics}.
As expected, tagged actions require a matching tag on the transition.
Expressions $e$ are
evaluated through a big step semantics $e \Downarrow v$ with $v$ a value, i.e.\
either $n \in \ntype$, $b \in \btype$, or  $q \in \qtype$.
We restrict to arithmetic and logical operations, and therefore omit the rules and assume that free variables are not
evaluated.
Note that the value of the qubits can only be observed through measurements, which alter such value and have a probabilistic outcome.
In particular, the expression $q = q'$ just compares two qubit names, and not their values.
In rule $\rulename{QOp}$, the superoperator $\mathcal{E}$ is applied to the qubits in $\tilde{q}$. Since $\tilde{q}$ can be smaller than the whole $\Sigma_\rho$,
we define $\mathcal{E}^{\tilde{q}}$ as
the superoperator that acts on the whole $\rho$ but ``ignores'' the qubits outside $\tilde{q}$.
More precisely, $\mathcal{E}^{\tilde{q}}$ is
obtained by composing
$(i)$ a suitable set of SWAP unitaries to bring the qubits $\tilde{q}$ in the first positions;
$(ii)$ the tensor product of the superoperator $\mathcal{E}$ with the identity on
untouched qubits on the right; and
$(iii)$ the inverse of the SWAP operators of point $(i)$ to recover the original order of qubits~\cite{lalireRelationsQuantumProcesses2006}.
In rule $\rulename{QMeas}$, given a measurement $\mathbb{M} = \{M_m\}$, for each $m$, $\mathcal{M}_m$ stands for the trace non-increasing superoperator such that $\mathcal{M}_m(\sigma) = M_m \sigma M_m^\dagger$, and $\mathcal{M}_m^{\tilde{q}}$ is defined as before.
In the rules $\rulename{ParL}$ and $\rulename{ParR}$ the parallel composition of a distribution $\Delta = \sum_{i \in I} \distelem{p_i}{\sconf{\rho_i, P_i}}$ and a process $Q$ is defined as $\sum_{i \in I} \distelem{p_i}{\sconf{\rho_i, P_i \parallel Q}}$, and similarly for the restriction $\Delta \setminus c$.
Notice that in $\rulename{ParL}$ we require that $P$ does not receive a qubit that is already owned by $Q$, and symmetrically in $\rulename{ParR}$.

\begin{example}\label{ex:quantumlottery2}
	The pLTS semantics of the quantum lottery $\proc{QL}$ from~\autoref{ex:quantumlottery} on quantum state $\dket{0}$ is in~\autoref{fig:ex-ql}, where
	$\proc{An}'$ stands for $t_4\tags\meas[01]{x}{y}.\proc{An}''$, and $\proc{An}''$ is $\ite{y = 0}{t_5\tags a!1.\nil_x}{t_6\tags b!1.\nil_x}$.
	To simplify the presentation, we draw $\xlongrightarrow{s\tags\mu}$ instead of $\xlongrightarrow{\mu}_s$.
	A scheduler for $\proc{QL}$ must decide $(i)$ which unitary
	$\proc{Pr}$ applies to the qubit $q$, either $X$ (i.e. $s = t_1$) or $H$ (i.e. $s = t_2$);
	and $(ii)$ if the qubit is sent to $\proc{An}$ or to the external environment (notice that the channel is not restricted), i.e. the available moves are with label $\tau$ and scheduler $(t_3, t_4)$, or with label $c!q$ and $s = t_3$.
	Intuitively, if $\proc{Pr}$ chooses $X$, then Bob will win, otherwise either Alice or Bob will win with the same probability.
	Note that at the beginning, in configuration $\conf{\ketbra{0}, \proc{QL}}$, the choice $t_4$ cannot be taken because the quantum state has no qubit apart from those in $\Sigma_{\proc{QL}}$.
\end{example}

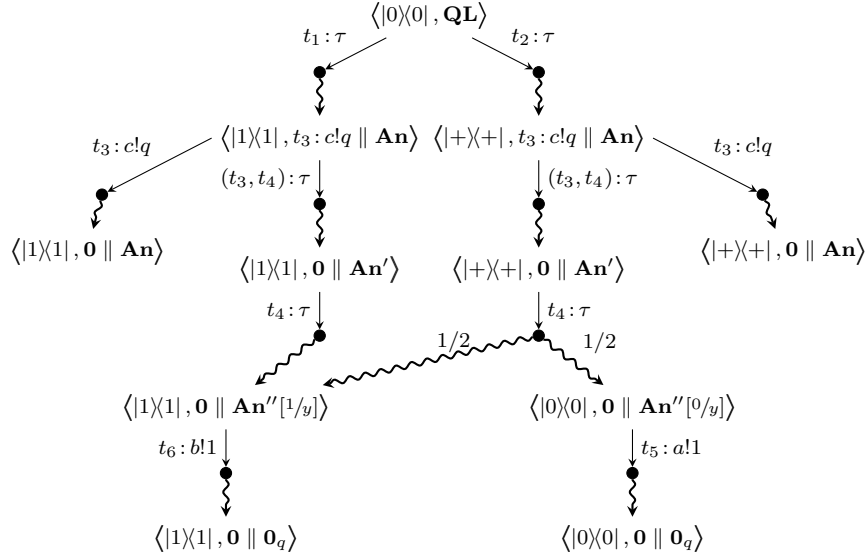
\begin{figure}[t]
	\centering
	\begin{tikzpicture}[node distance=5mm]
		\node [conf] (c1) at (0, 0) {$\conf{\ketbra{0}, \proc{QL}}$};
		\node [psplit, below left=4mm and 5mm of c1] (p12) {};
		\node [psplit, below right=4mm and 5mm of c1] (p13) {};
		\node [conf, below= of p12] (c2) {$\conf{\ketbra{1}, t_3\tags c!q \parallel\proc{An}}$};
		\node [conf, below= of p13] (c3) {$\conf{\ketbra{+}, t_3\tags c!q \parallel\proc{An}}$};
		\node [psplit, below=of c2] (p2) { };
		\node [psplit, below left=4mm and 14mm of c2] (p2') { };
		\node [psplit, below=of c3] (p3) { };
		\node [psplit, below right=4mm and 14mm of c3] (p3') { };
		\node [conf, below left=4mm and -10mm of p2'] (c4) {$\conf{\ketbra{1}, \nil \parallel\proc{An}}$};
		\node [conf, below=of p2] (c5) {$\conf{\ketbra{1}, \nil \parallel\proc{An}'}$};
		\node [conf, below=of p3] (c6) {$\conf{\ketbra{+}, \nil \parallel\proc{An}'}$};
		\node [conf, below right=4mm and -10mm of p3'] (c7) {$\conf{\ketbra{+}, \nil \parallel\proc{An}}$};
		\node [psplit, below=of c5] (p5) { };
		\node [psplit, below=of c6] (p6) { };
		\node [conf, below left=6mm and -3mm of p5] (c8) {$\conf{\ketbra{1}, \nil \parallel\proc{An}''[\sfrac{1}{y}]}$};
		\node [conf, below right=6mm and -3mm of p6] (c9) {$\conf{\ketbra{0}, \nil \parallel\proc{An}''[\sfrac{0}{y}]}$};
		\node [psplit, below=of c8] (p8) { };
		\node [psplit, below=of c9] (p9) { };
		\node [conf, below=of p8] (c8f) {$\conf{\ketbra{1}, \nil \parallel \nil_q}$};
		\node [conf, below=of p9] (c9f) {$\conf{\ketbra{0}, \nil \parallel \nil_q}$};

		\draw [arrow] (c1) -- (p12) node[weight,above left] {$t_1\tags\tau$};
		\draw [arrow] (c1) -- (p13) node[weight,above right] {$t_2\tags\tau$};
		\draw [parrow] (p12) -- (c2) node[weight,above left=.3mm and 1mm] { };
		\draw [parrow] (p13) -- (c3) node[weight,above left=.3mm and 1mm] { };
		\draw [arrow] (c2) -- (p2) node[weight,left] {$(t_3, t_4)\tags\tau$};
		\draw [arrow] (c3) -- (p3) node[weight,right] {$(t_3, t_4)\tags\tau$};
		\draw [arrow] (c2.west) -- (p2') node[weight,above left] {$t_3\tags c!q$};
		\draw [arrow] (c3.east) -- (p3') node[weight,above right] {$t_3 \tags c!q$};
		\draw [parrow] (p2') -- (c4) node[weight,left] { };
		\draw [parrow] (p2) -- (c5) node[weight,above left] { };
		\draw [parrow] (p3) -- (c6) node[weight,right] { };
		\draw [parrow] (p3') -- (c7) node[weight,above right] { };
		\draw [arrow] (c5) -- (p5) node[weight,left] {$t_4\tags\tau$};
		\draw [arrow] (c6) -- (p6) node[weight,right] {$t_4\tags\tau$};
		\draw [parrow] (p5) -- (c8) node[weight,right] { };
		\draw [parrow] (p6) -- (c8) node[weight,above right] {$1/2$};
		\draw [parrow] (p6) -- (c9) node[weight,above right] {$1/2$};
		\draw [arrow] (c8) -- (p8) node[weight,left] {$t_6\tags b!1$};
		\draw [arrow] (c9) -- (p9) node[weight,right] {$t_5\tags a!1$};
		\draw [parrow] (p8) -- (c8f) node[weight,right] { };
		\draw [parrow] (p9) -- (c9f) node[weight,above right] { };
	\end{tikzpicture}
	\caption{Semantics of the quantum lottery process.}
	\label{fig:ex-ql}
\end{figure}

It is useful to consider a pLTS as a simple LTS between (sub-)distributions.
First, we add a ``deadlock'' configuration $\bot$, allowing distributions to evolve in sub-distributions if a move is not legal for some elements of the support of the distribution~\cite{dengbisimulations2018}.
Formally, we let $\confbot = \confset \cup \{ \bot \}$ and define $bot(\longrightarrow) \in \confbot \times \actset \times \choiceset \times \dist{\confbot}$ as the least relation such that $bot(\longrightarrow)\ \supseteq\ \longrightarrow$ and $\mathcal{C}\ bot(\xlongrightarrow{\mu}_{s})\ \overline{\bot}$ if there is no $\Delta$ such that $\mathcal{C} \xlongrightarrow{\mu}_{s} \Delta$.
We extend typing to $\confbot$ by imposing $(\Sigma, \Sigma') \vdash \bot$ for any $\Sigma$ and $\Sigma'$.
Since an element $\Delta$ of $D(\confbot)$ represents a sub-distribution, we call \emph{mass} the probability of all the configurations in $\Delta$ different from $\bot$: $\mass{\Delta} = 1 - \Delta(\bot)$.

Then, we lift the semantics to deal with distributions of schedulers and distributions of configurations.
Consider
$\lift_{\confbot}(\lift_{\choiceset}(bot(\longrightarrow))) : \dist{\confbot} \times \actset \times \dist{\choiceset} \times \dist{\confbot}$, where we first lift the transition relation to distributions of schedulers, hence allowing for randomized schedulers~\cite{Segala95}, and then lift over the input configurations.
When clear from the context, we will adopt $\longrightarrow$ also as the symbol for the lifting of the semantics to distributions, and we will write $\longrightarrow_s$ for $\longrightarrow_{\overline{s}}$ with $\overline{s}$ the randomized scheduler that behaves as $s$ with probability 1.
To avoid clashing notation, we will use $\delta$ instead of $\Delta$ for distributions of schedulers, e.g. in the transition $\confel \tauarrow_\delta \Delta$.

\begin{example}\label{ex:quantumlottery3}
	Recall \autoref{ex:quantumlottery2}. In the lifted semantics, it is possible to define a randomized scheduler choosing the unitary by tossing a fair coin, i.e. obtaining the following sequence of reductions
	\begin{align*}
		 & \sconf{\ketbra{0}, \proc{QL}} \xlongrightarrow{\tau}_{\overline{t_1} \psum{1/2} \overline{t_2}}
		\sconf{\ketbra{1}, c!q \parallel\proc{An}}
		\!\psum{\frac{1}{2}}\!
		\sconf{\ketbra{+}, c!q \parallel\proc{An}}                                                         \\
		 & \ \xlongrightarrow{\tau}_{\overline{(t_3, t_4)}}
		\sconf{\ketbra{1}, \nil \parallel\proc{An}'}
		\!\psum{\frac{1}{2}}\!
		\sconf{\ketbra{+}, \nil \parallel\proc{An}'}                                                       \\
		 & \ \xlongrightarrow{\tau}_{\overline{t_4}}
		\sconf{\ketbra{1}, \nil \parallel\proc{An}''[\sfrac{1}{y}]}
		\!\psum{\frac{3}{4}}\!
		\sconf{\ketbra{0}, \nil \parallel\proc{An}''[\sfrac{0}{y}]}
		\xlongrightarrow{a!1}_{\overline{t_5}}
		\overline{\bot}\!
		\psum{\frac{3}{4}}\!
		\sconf{\ketbra{0}, \nil \parallel \nil_{q}}
	\end{align*}
	Note that in the last step we end up in a sub-distribution, since the choice $t_5$ with action $a!1$ is not available for $\sconf{\ketbra{1}, \nil \parallel\proc{An}''[\sfrac{1}{y}]}$, which transitions to $\overline{\bot}$.
\end{example}

It is worth noting that the typing is preserved by $\tau$-transitions.
\begin{restatable}[Typing Preservation]{theorem}{typingpreservation}\label{thm:typepreservation}
	If $(\Sigma_\rho, \Sigma_P) \vdash \conf{\rho, P}$ and $\conf{\rho, P} \xlongrightarrow{\tau}_{s} \Delta$ then
	$(\Sigma_\rho, \Sigma_P) \vdash \Delta$.
\end{restatable}


As hinted by the previous theorem, type preservation does not hold in general.
However, the updated typing context is uniquely determined by the transition label, and it is still a subset of the qubits available in the quantum state.
\begin{restatable}[Typing Quasi-Preservation]{theorem}{quasitypingpreservation}\label{thm:quasipreservation}
	%
	Let $\Sigma \subseteq \qtype$ and $\mu \in Act$. Then there exists $\Sigma'$ such that for all
	$(\Sigma'', \Sigma) \vdash \iconf$ and $\Delta \in \dist{\confset}$ with $\iconf \xlongrightarrow{\mu}_s \Delta$
	it holds $(\Sigma'', \Sigma') \vdash \Delta$.

	%


	%
	%
	%
\end{restatable}


Hereafter, we make the standard assumption that
processes are tagged in a deterministic way, meaning that the evolution of distributions is uniquely determined by the scheduler~\cite{chatzikokolakisMakingRandomChoices2007,chatzikokolakisBisimulationDemonicSchedulers2009}.
\begin{assumption}
	For any distribution $\Delta \in \dist{\confbot}$, scheduler distribution $\delta$ and action $\mu$, if $\Delta \xlongrightarrow{\mu}_\delta \Delta'$ and $\Delta \xlongrightarrow{\mu}_\delta \Delta''$, then $\Delta' = \Delta''$.
\end{assumption}

\section{Behavioural Equivalence}\label{sec:db}

We start by defining a saturated bisimilarity \emph{\`a la}~\cite{bonchigeneral2014}, a natural notion of behavioural equivalence that pairs systems when they are indistinguishable for any observer.
We show that schedulers guarantee adherence with the prescriptions of quantum theory, as contexts cannot choose their move based on the unknown states of unmeasured qubits, a problem previously highlighted in~\cite{ceragioliQuantumBisimilarityBarbs2024}.
Finally, we give an equivalent characterization in terms of labelled bisimulations.

\subsection{Saturated Bisimilarity}


In saturated bisimilarities,
contexts $B[\blank]$ are processes with a typed hole that play the
role of process-discriminating observers.
\begin{definition}
	A context $B[\blank]_{\Sigma}$ is generated by the production
	$B[\blank]_{\Sigma} \Coloneqq [\blank]_{\Sigma} \parallel P$, typed
	according to the rules in Appendix~\ref{sec:typeappendix},
	and to the following one
	\[
		\infer[\rulename{Hole}]{\Sigma' \vdash [\blank]_{\Sigma} \parallel P}
		{\Sigma' \setminus \Sigma \vdash P  & \Sigma \subseteq \Sigma'}
	\]
\end{definition}
A process $P$ is
applied to contexts by replacing the hole with $P$.
Intuitively, a context $\Sigma' \vdash B[\blank]_{\Sigma}$ is a function that
given a process $P$ returns a process $B[P]$ obtained by replacing $P$ for $[\blank]$,
where $\Sigma$ is the typing context of the valid inputs and $\Sigma'$ the one
of the outputs.
Note that a context can own some qubits and each qubit cannot be referred to in both $P$ and $B[\blank]$.
We apply $\Sigma' \vdash B[\blank]_{\Sigma}$ to configurations $(\Sigma_\rho, \Sigma_P) \vdash \conf{\rho, P}$ obtaining $(\Sigma_\rho, \Sigma') \vdash \conf{\rho, B[P]}$ when
$\Sigma' \subseteq \Sigma_\rho$ and $\Sigma = \Sigma_P$, i.e.\ when the qubits referred by $B[\blank]$ are defined in $\rho$ and the process $P$ is as prescribed by $B[\blank]$.
We write $B[\conf{\rho, P}]$ for $\conf{\rho, B[P]}$, $B[\bot]$ for $\bot$, and
$B[\Delta]$ for the distribution obtained by applying $B[\blank]$ to the support of $\Delta$.
It is trivial to show that if $\Delta$ and $\Theta$ are typed by the same typing context, then $B[\Delta]$ is defined if and only if $B[\Theta]$ is defined.

A saturated bisimulation is a relation over distributions where related pairs must have the same type and mass, and must reduce in related distributions under every possible context.
\begin{definition}[Saturated Bisimilarity]
	A relation $\rel \subseteq \dist{\confbot} \times \dist{\confbot}$ is a
	\emph{saturated bisimulation} if $\Delta\,\rel\,\Theta$ implies
	$(\Sigma, \Sigma') \vdash \Delta$ and $(\Sigma, \Sigma') \vdash \Theta$ for some $\Sigma, \Sigma'$, $\mass{\Delta} = \mass{\Theta}$ and
	for any context $B[\blank]_{\Sigma'}$ it holds that
	\begin{itemize}
		\item whenever $B[\Delta] \tauarrow_\delta \Delta'$, there exists $\Theta'$
		      such that $B[\Theta] \tauarrow_\delta \Theta'$ and $\Delta'\;\rel\;\Theta'$;
		\item whenever $B[\Theta] \tauarrow_\delta \Theta'$, there exists $\Delta'$
		      such that $B[\Delta] \tauarrow_\delta \Delta'$ and $\Delta'\;\rel\;\Theta'$.
	\end{itemize}

	Let \emph{saturated bisimilarity}, denoted $\sim_s$, be the largest saturated bisimulation.
\end{definition}
Notice that $\sim_s$ is a congruence with respect to $\parallel$ by definition.
We now compare our proposed bisimilarity with the prescriptions of quantum theory.

\subsection{Assessment of Saturated Bisimilarity}\label{sec:ass}

As a first result, we notice that $\sim_s$ is a linear relation, meaning that the convex combination of bisimilar distributions yields bisimilar distributions.
\begin{restatable}{theorem}{linearityCongruence}\label{thm:linearity}
	If $\Delta_i \sim_s \Theta_i$ for $i\! =\! 1,2$ and $p\in [0,1]$ then $\Delta_1 \psum p \Delta_2 \sim_s \Theta_1 \psum p \Theta_2$.
\end{restatable}

The second property we address is purely quantum, and lifts to lqCCS the indistinguishability relations between quantum states of~\autoref{thm:qind}.

This property has been originally defined in~\cite{ceragioliQuantumBisimilarityBarbs2024}, and states that the same process, acting on two different but indistinguishable mixed quantum states, exhibits a behaviour that cannot be distinguished by any observer, therefore yielding bisimilar distributions.

%
\begin{restatable}{theorem}{propertyA}\label{thm:propertyA}
	If $\sum_i p_i \cdot \rho_i = \sum_j q_j \cdot \rho_j$ then
	$\sum_{i} p_i \bullet \sconf{\rho_i, P} \sim_s \sum_{j} q_j \bullet \sconf{\sigma_j, P}.$
\end{restatable}
\begin{proofsketch}
	We prove by induction $\Delta = \sconf{\rho , P} \psum{p} \sconf{\sigma, P} \sim_s \sconf{\rho \psum{p} \sigma, P}= \Theta$, and the theorem follows by transitivity.
	First, we show that $\Delta$ replicates the moves of $\Theta$ by the linearity of superoperators, i.e. $\E(\rho) \psum{p} \E(\sigma) = \E(\rho \psum{p} \sigma)$, and similarly for measurements.
	Then, we show that $\Theta$ simulates $\Delta$. 
	Here, the presence of schedulers is key: it forbids 
	$\Delta$ from combining different non-deterministic choices to perform a move that would not be available to $\Theta$ (see~\autoref{ex:broken-nondet}). 
\end{proofsketch}


This result directly derives from the use of schedulers, and it is not common in quantum versions of CCS: it holds only for specific distributions in~\cite{ceragioliQuantumBisimilarityBarbs2024} and in~\cite{dengbisimulations2018}.
To see the role of schedulers, consider the following example about indistinguishable qubit sources, formalizing the intuition of~\autoref{fig:ex-zopm}.

\begin{example}\label{ex:broken-nondet}
	Consider a pair of non-biased random qubit sources, the first sending a qubit in state $\kz$ or
	$\ko$, the second in state $\kpl$ or $\km$.
	Quantum theory prescribes that these two sources cannot be
	distinguished by any observer, as the received qubit behaves the same~\cite{nielsenquantum2010}.
	Indeed, the
	(mixed) states of the qubits sent by both
	sources are represented by the 
	density operator $\frac{1}{2}\I$.
	Fittingly, the lqCCS encodings of these sources
	 are bisimilar by \autoref{thm:propertyA}: $\Delta_{01} \sim_s \Delta_{\pm}$, for
	\begin{align*}
		\Delta_{01}  &= \singleton{\conf{\dket{0}, t_0\tags c!q}} \psum{1/2} \singleton{\conf{\dket{1}, t_0\tags c!q}} \sim_s 
		\singleton{\lrconf{\frac{1}{2}\I, t_0\tags c!q}}\\
		\Delta_{\pm} &= \singleton{\conf{\dket{+}, t_0\tags c!q}} \psum{1/2} \singleton{\conf{\dket{-}, t_0\tags c!q}} \sim_s
				\singleton{\lrconf{\frac{1}{2}\I, t_0\tags c!q}}
	\end{align*}
	By contrast, assume the \emph{unscheduled semantics} that ignores tags,
	defined as
		$\xlongrightarrow{\mu}_u\ = \lift_{\confbot}(\bigcup_{s} bot(\xlongrightarrow{\mu}_s)) : 
		\dist{\confbot} \times \actset \times \dist{\confbot}$, 
	and let the \emph{unscheduled bisimilarity} $\sim_{us}$ be the saturated bisimilarity for this transition relation.
	The unscheduled bisimilarity erroneously discriminates the two sources.
	To see that $\Delta_{01} \not\sim_{us} \Delta_{\pm}$, take $B[\blank] = [\blank] \parallel t_1\tags c?x.(P+Q)$ where
	\begin{gather*}
		P =  t_2\tags \meas[{01}]{x}{y} .R\ \text{, and }\
		Q =  t_3\tags \meas[{\pm i}]{x}{y} .R\text{, with}\\
		R =  (\ite{y=0}{t_4\tags \tau}{\nil}) \parallel \nil_{x}
	\end{gather*}
	Notice that $P$ and $Q$ perform different measurements and
	make the outcome observable by enabling a $\tau$-transition only when $y = 0$.

	Let $\confel_{\psi}$ be $\conf{\dket{\psi}, P+Q}[\sfrac{q}{x}]$, then $B[\Delta_{01}]$ reduces to $\Delta_{01}' = \singleton{\confel_0} \psum{1/2} \singleton{\confel_1}$, and $B[\Delta_{\pm}]$ can only reduce to $\Delta_{\pm}' = \singleton{\confel_+} \psum{1/2} \singleton{\confel_-}$ to match this move.
	The following moves are available for $\confel_0$ and $\confel_1$
	\begin{align*}
		{\confel_0} & \xlongrightarrow{\tau}_{\singleton{t_2}}     \overline{\conf{\ketbra{0}, R[\sfrac{0}{y}]}} = \Delta_0,           \\
		{\confel_1} & \xlongrightarrow{\tau}_{\singleton{t_3}}     \overline{\conf{\ketbra{i}, R[\sfrac{0}{y}]}} \psum{1/2}
		\overline{\conf{\ketbra{-i}, R[\sfrac{1}{y}]}} = \Delta_1
	\end{align*}
	Consider now the convex combination of the two distributions above, $\Delta_{01}'' = \Delta_0 \psum{1/2} \Delta_1$,
	and notice that ${\confel_0} \xlongrightarrow{\tau}_{u} \Delta_0$ and ${\confel_1} \xlongrightarrow{\tau}_{u} \Delta_1$.
	In the unscheduled semantics, $\Delta_{01}'$ can mix the two choices and reduce $\confel_0$ with $t_2$ and $\confel_1$ with $t_3$, formally,
		$\Delta_{01}' \xlongrightarrow{\tau}_{u} \Delta_{01}''$.
	
%
	Finally, note that $\Delta_{01}'' \longrightarrow_{\overline{t_4}} \Theta_{01}$ with $\mass{\Theta_{01}} = 3/4$.

	Consider now $\Delta_{\pm}'$, all the available moves for $\confel_{+}$ and $\confel_{-}$ follow
	\begin{align*}
		{\confel_{+}} / \confel_{-} & \xlongrightarrow{\tau}_{\singleton{t_2}} \overline{\conf{\ketbra{0}, R[\sfrac{0}{y}]}} \psum{1/2}
		\overline{\conf{\ketbra{1}, R[\sfrac{1}{y}]}},                                                                                  \\
		{\confel_{+}} / \confel_{-} & \xlongrightarrow{\tau}_{\singleton{t_3}} \overline{\conf{\ketbra{i}, R[\sfrac{0}{y}]}} \psum{1/2}
		\overline{\conf{\ketbra{-i}, R[\sfrac{1}{y}]}}.
	\end{align*}
	It is easy to check that $\Delta_{\pm}'$ cannot replicate the behaviour of $\Delta_{01}'$: for any
	choice of $t_2$ and $t_3$ it will reduce to a distribution $\Theta_{\pm}$ with $\mass{\Theta_{\pm}} = 1/2$.
	
	Schedulers solve this issue by forbidding to combine moves labelled by different choices:
	$\Delta_{01}'$ cannot mix the two choices in our scheduled semantics, and this allows us to prove~\autoref{thm:propertyA}.
	
\end{example}

\autoref{ex:broken-nondet} is paradigmatic, where different mixtures of quantum states are discriminated because the moves are chosen
according to the value of some received qubit, which in theory should be unknown.
Our schedulers do not allow processes to replicate this behaviour, forcing distributions to make a reasonable choice based on classically determined tags.

\subsection{Labelled Bisimilarity}\label{sec:siml}
In this section we give an equivalent labelled characterization of our saturated bisimilarity, thus making explicit the observable properties of lqCCS processes.

We need some auxiliary definitions.
First, given a distribution $\Delta$ of lqCCS processes, we let $qs(\Delta)$ be its quantum state $\sum_{\langle{\rho, P}\rangle \in \support{\Delta}} \Delta(\conf{\rho, P}) \cdot \rho$.
Notice that $qs(\Delta)$ is a
density operator in $\sdensity{\hilb}$, with $\tr(qs(\Delta)) = \mass{\Delta}$.

As a second ingredient, we allow applying superoperators to configurations and distributions.
%
Given $\E \in \soset{\hilbert}$, we define
$\nE : \dist{\confset} \to \dist{\confset}$
as
$$\nE\left(\sum_i \distelem{p_i}{\conf{\rho_i, P_i}}\right) = \sum_i \distelem{\frac{p_i\cdot tr(\E(\rho_i))}{p_{\E}}}\overline{\left\langle {\frac{\E(\rho_i)}{\tr(\E(\rho_i))}, P_i}\right\rangle}$$
where $p_{\E} = \sum_i p_i\cdot tr(\E(\rho_i))$.
Roughly, if $\E$ is trace-preserving, $\nE$ just applies $\E$ in the configurations.
If $\E$ is trace non-increasing, $\nE$ applies $\E$ to each $\rho_i$ and normalizes the resulting matrix; the new weight of the $i$-th configuration is the conditioned probability of being in $i$ knowing that $\E$ happened.
We extend $\nE$ to $\dist{\confbot}$ imposing that $\nE(\Delta \psum{p} \bot) = \nE(\Delta) \psum{p} \bot$ if $\mass{\Delta} = 1$.
Given $\Sigma, \Sigma' \vdash \Delta$, if $\E$ is defined on $\tilde{q}$, which are only some of the qubits in $\Sigma$, we write $\nE(\Delta)$ for $\nE^{\tilde{q}}(\Delta)$ (recall that $\nE^{\tilde{q}}$ extends $\nE$ by tensoring it with the identity).

We can now define our labelled bisimilarity.
In addition to the usual conditions about labelled transitions, labelled bisimulations are required to satisfy two additional conditions:
paired distributions must share the same \emph{environment}, i.e., the portion of the quantum state that is immediately visible to the context; and the relation
must be closed for the application of normalized superoperators over the qubits of the environment.
This superoperator-closure is needed for proving that labelled bisimilarity has the same observing power of saturated bisimilarity,
as a context $B[\blank] = [\blank] \parallel Q$ can read and modify the qubits of the environment.
\begin{definition}[Labelled Bisimilarity]\label{def:labBisim}
	A relation $\rel \subseteq \dist{\confbot} \times \dist{\confbot}$ is a
	\emph{labelled bisimulation} if $\Delta\,\rel\,\Theta$ implies
	$(\Sigma, \Sigma') \vdash \Delta$ and $(\Sigma, \Sigma') \vdash \Theta$ for some $\Sigma, \Sigma'$, and it holds that
	\begin{itemize}
		\item $\tr_{\Sigma'}(qs(\Delta)) = \tr_{\Sigma'}(qs(\Theta))$;
		\item $\nE(\Delta)\ \rel\ \nE(\Theta)$ for any superoperator $\mathcal{E} \in \soset{\hilbert_{\Sigma \setminus \Sigma'}}$;
		\item whenever $\Delta \xlongrightarrow{\mu}_\delta \Delta'$, there exists $\Theta'$
		      such that $\Theta \xlongrightarrow{\mu}_\delta \Theta'$ and $\Delta'\;\rel\;\Theta'$;
		\item whenever $\Theta \xlongrightarrow{\mu}_\delta \Theta'$, there exists $\Delta'$
		      such that $\Delta \xlongrightarrow{\mu}_\delta \Delta'$ and $\Delta'\;\rel\;\Theta'$.
	\end{itemize}

	Let \emph{labelled bisimilarity}, denoted $\sim_l$, be the largest labelled bisimulation.
\end{definition}

Note that superoperators acting over qubits that are not in the process may affect some of the qubits of the process too, due to entanglement.
\begin{example}
	Let $\confel = \conf{ \ketbra{\Phi^+}, \meas[01]{q_2}{x}.c!x.\nil_{q_2}}$ for $(\{ q_1, q_2 \}, \{ q_2 \}) \vdash \confel$.
	Consider the trace non-increasing superoperator $\mathcal{E}^{q_1} = \{\ketbra{0}\}$ projecting the state of the first qubit (in the environment) to the value $\ketbra{0}$.
	The result of applying $\nE^{q_1}$ to $\confel$ is $\conf{ \ketbra{00}, \meas[01]{q_2}{x}.c!x.\nil_{q_2}}$, where also the second qubit held by the process is updated.
\end{example}

Since the additional requirements of $\sim_l$ are inspired by discriminating contexts, it is to be expected that $\sim_s \ \subseteq \ \sim_l$.
Indeed, the converse is also true. 
\begin{restatable}{theorem}{corrcompl}\label{thm:corrcompl}
	For any $\Delta$ and $\Theta$, $\Delta \sim_l \Theta$ if and only if $\Delta \sim_s \Theta$.
\end{restatable}
\begin{proofsketch}
	First, we prove that $\sim_s$ is a labelled bisimulation: we assume $\Delta \sim_s \Theta$ and show that if the conditions of~\autoref{def:labBisim} do not hold, then there exists a distinguishing context $B[\blank]$ contradicting our assumption.
	If the visible quantum states differ, then $B[\blank]$ just performs a measurement.
	If $\nE(\Delta) \not\sim \nE(\Theta)$, then $B[\blank]$ applies a trace non-increasing superoperator on visible qubits.
	Finally, if $\Delta$ performs a labelled transition (e.g. an input action) that cannot be matched by $\Theta$, then $B[\blank]$ perform the ``dual'' transition (e.g. an output action).

	Then, we prove that $\sim_l$ is a saturated bisimulation.
	Given $\Delta \sim_l \Theta$, and a generic context $B[\blank]$, we consider all the three possible transitions for $B[\Delta]$:
	i) $\Delta$ moves with a $\tau$ action, thus also $\Theta$ does;
	ii) $B[\blank]$ and $\Delta$ synchronize, meaning that $\Delta$ moves with a visible action, and so does $\Theta$;
	iii) $B[\blank]$ moves, possibly measuring or modifying the visible qubits, but thanks to the first two bullet points of \autoref{def:labBisim}, $B[\Delta]$ and $B[\Theta]$ express the same behaviour.
	%
\end{proofsketch}

We conclude with two real-world examples, quantum teleportation~\cite{qteleportation} and superdense coding~\cite{SDC}.
\begin{example}
	The objective of quantum teleportation is to allow Alice to send quantum information to Bob 
	without a quantum channel.
	Alice and Bob must have each one of the qubits of an entangled pair $\ket{\Phi^+}$.
	The protocol works as follows: Alice performs a fixed set of unitaries to the qubit to transfer and to her part of the entangled pair;
	then, she measures the qubits and sends the classical outcome to Bob, which applies different unitaries to his own qubit according to the received information.
	In the end, the qubit of Bob will be in the state of Alice's one, and the entangled pair is discarded.

	Consider the following encoding of the protocol $\proc{Tel} = (\proc{A} \parallel \proc{B}) \setminus c$, where we assume that Alice ($\proc{A}$) and Bob ($\proc{B}$) already share an entangled pair ($q_1,q_2$) (we write $(n)_2$ to stress that $n$ is in binary representation)
	\begin{align*}
		\proc{A}    & = t \tags \unitary{CNOT}{q_0,q_1}.t \tags \unitary{H}{q_0}.t \tags \tmeas{q_0,q_1}{x}.(t \tags c!x \parallel \nil_{q_0,q_1}) \\
		\proc{B}    & = t' \tags c?y.\ \ite{y = (00)_2}{t' \tags \I(q_2).t' \tags \mathit{out}!q_2                                                       \\
		            & \quad\;\;\,}
		{( \ite{y = (01)_2}{
		t' \tags \unitary{X}{q_2}.t' \tags  \mathit{out}!q_2                                                                                             \\
		            & \quad\;\;\,}
		{( \ite{y = (10)_2}{t' \tags \unitary{Z}{q_2}.t' \tags \mathit{out}!q_2                                                                          \\
		            & \quad\;\;\,}
		{t' \tags \unitary{ZX}{q_2}.t' \tags \mathit{out}!q_2})}) }                                                                                      \\
		\proc{Spec} & = t\tags\unitary{SWAP}{q_0,q_2}.t\tags\tau.t\tags\tau.(t,t')\tags\tau.t'\tags\tau.(t'\tags\text{out}!q_2 \parallel \nil_{q_0,q_1})
	\end{align*}
	where $\mathbb{M}_{23}^{01}$ is the two-qubit measurement in the basis $\{ \ket{00}, \ket{01}, \ket{10}, \ket{11}\}$.
	We let $\Delta = \singleton{\conf{\dket{\psi}\otimes\dket{\Phi^+}, \proc{Tel}}}$, with $\Theta = \singleton{\conf{\dket{\psi}\otimes\dket{\Phi^+}, \proc{Spec}}}$ its specification, for $\kp = \alpha\kz + \beta\ko$, and sketch the proof for $\Delta \sim_{l} \Theta$ below.
	Note that $\proc{Spec}$ simply states that, after the tagged operations, the state of qubits is swapped and the correct state is communicated over the expected channel.

	Since there is no qubit in $\dket{\psi\Phi^+}$ apart from the ones in $\Sigma_{\proc{Tel}}$, the environment of the two distributions is trivially the same, and no superoperator is to be considered.
	The evolution of $\Delta$ and $\Theta$ is exactly the same until a send over the unrestricted channel $out$ is reached.
	The relevant steps are
	\begin{align*}
		 & \Delta \xlongrightarrow{\tau}_t^3
		\sum\nolimits_{n = 0}^3 \distelem{\frac{1}{4}}{\sconf{\dket{n} \otimes \dket{\psi_n}, (t\tags m!n \parallel \nil_{q_0,q_1} \parallel \proc{B}) \setminus c}} \\
		 & \quad\xlongrightarrow{\tau}_{(t,t')}\xlongrightarrow{\tau}_{t'} \ \Delta' =
		\sum\nolimits_{n = 0}^3 \distelem{\frac{1}{4}}{\sconf{\dket{n} \otimes \dket{\psi}, (\nil_{q_0,q_1} \parallel t'\tags \mathit{out}!q_2) \setminus c}}        \\
		 & \Theta  \xlongrightarrow{\tau}_t^3  \xlongrightarrow{\tau}_{(t,t')}\xlongrightarrow{\tau}_{t'}  \
		\Theta' = \sconf{\dket{\Phi^+}\otimes\dket{\psi}, t'\tags\mathit{out}!q_2 \parallel \nil_{q_0,q_1}}
	\end{align*}
	where $\ket{\psi_{0}} = \kp$, $\ket{\psi_{1}} = \beta\kz + \alpha\ko$, $\ket{\psi_{2}} = \alpha\kz - \beta\ko$, $\ket{\psi_{3}} = \beta\kz - \alpha\ko$,
	and where, abusing notation, we use $\ket{0} = \ket{00}, \ket{1} = \ket{01}, \ket{2} = \ket{10}$ and $\ket{3} = \ket{11}$ when speaking of pairs of qubits.
	The last move is then a send on the channel $out$ for both distributions
	\begin{align*}
		\Delta' \xlongrightarrow{out!q_2}_{t'}
		 & \ \Delta'' = \sum\nolimits_{n = 0}^3 \distelem{\frac{1}{4}}{\sconf{\dket{n} \otimes \dket{\psi}, (\nil_{q_0,q_1})}} \\
		\Theta' \xlongrightarrow{out!q_2}_{t'}
		 & \ \Theta'' = \sconf{\dket{\Phi^+}\otimes\dket{\psi}, \nil_{q_0,q_1}}
	\end{align*}
	The bisimilarity then is easily checked as $\Delta''$ and $\Theta''$ are in deadlock, and $\tr_{q_0,q_1}(qs(\Delta'')) = \dket{\psi} = \tr_{q_0,q_1}(qs(\Theta''))$.


\end{example}

\begin{example}
	Assume Alice and Bob have each a qubit of $\ket{\Phi^+}$.
	The protocol allows Alice to communicate a
	two-bit integer to Bob by sending her single qubit.

	The protocol is as follows: Alice chooses an integer in $[0,3]$ and encodes it by applying suitable transformations to her qubit, which is then sent to Bob;
	Bob receives the qubit and decodes
	it
	by performing CNOT and $\text{H} \otimes \I$ on the pair of qubits (the received qubit and his original one).
	Finally, he measures the qubits in the standard basis, recovering the integer chosen by Alice.

	We consider the following encoding of the protocol 	$\proc{SDC} = (\proc{A} \parallel \proc{B}) \setminus c$.
	\begin{align*}
		\proc{A}  =\
		                & (t_0\tags \I(q_0). t\tags c!q_0) +  (t_1\tags X(q_0). t\tags c!q_0)                                                           \\
		                & + (t_2\tags Z(q_0). t\tags c!q_0) +  (t_3\tags ZX(q_0). t\tags c!q_0)                                                         \\
		\proc{B}    =\  & t'\tags c?x. t'\tags \unitary{CNOT}{x, q_1}.t'\tags \unitary{H}{x}. t'\tags \tmeas{x,q_1}{y}.t'\tags out!y.\nil_{x,q_1} \\
		\proc{Spec}   =\
		                & t_0\tags\tau. (t,t')\tags\tau. t'\tags\tau .t'\tags\tau . t'\tags\tau . t'\tags out!0. \nil_{q_0,q_1}\ +                \\
		                & t_1\tags\tau. (t,t')\tags\tau. t'\tags\tau .t'\tags\tau . t'\tags\tau . t'\tags out!1. \nil_{q_0,q_1}\ +                \\
		                & t_2\tags\tau. (t,t')\tags\tau. t'\tags\tau .t'\tags\tau . t'\tags\tau . t'\tags out!2. \nil_{q_0,q_1}\ +                \\
		                & t_3\tags\tau. (t,t')\tags\tau. t'\tags\tau .t'\tags\tau . t'\tags\tau . t'\tags out!3. \nil_{q_0,q_1}
	\end{align*}

	We let $\Delta = \singleton{\conf{\dket{\Phi^+}, \proc{SDC}}}$, with $\Theta = \singleton{\conf{\dket{\Phi^+}, \proc{Spec}}}$ its specification.
	Note that $\proc{Spec}$ states that the choice of the initial unitary determines the final value to be communicated.

	Since there is no qubit in $\dket{\Phi^+}$ apart from the ones in $\Sigma_{\proc{SDC}}$,
	the environment of the two distributions is trivially the same, and no superoperator is to be considered.
	The evolution of $\Delta$ and $\Theta$ are as follows, for $n = 0,1,2,3$
	\begin{align*}
		 & \Delta \xlongrightarrow{\tau}_{t_n} \xlongrightarrow{\tau}_{(t,t')}\xlongrightarrow{\tau}_{t'}^2
		\sconf{\dket{n}, \tmeas{q_0,q_1}{y}.t'\tags out!y.\nil_{q_0,q_1} \setminus c}                       \\
		 & \quad\xlongrightarrow{\tau}_{t'} \Delta_n = \sconf{\dket{n}, t'\tags out!n.\nil_{q_0,q_1} \setminus c} \\
		 & \Theta  \xlongrightarrow{\tau}_{t_n} \xlongrightarrow{\tau}_{(t,t')}\xlongrightarrow{\tau}_{t'}^3  \
		\Theta_n = \sconf{\dket{\Phi^+}, t'\tags out!n.\nil_{q_0,q_1} \setminus c}
	\end{align*}
	where we use $\ket{n}$ for the two qubits binary representation of the natural number $n \in \{0,1,2,3\}$, namely $\ket{0} = \ket{00}, \ket{1} = \ket{01}, \ket{2} = \ket{10}$ and $\ket{3} = \ket{11}$.

	Finally, note that $\Delta_n \sim_{l} \Theta_n$ for any $n$, because both distributions $\Delta_n$ and $\Theta_n$ sends the number $n$ on the channel $out$, and reduce to deadlock distributions with empty quantum environment.
\end{example}


%

\section{Conclusions and Future Work}\label{sec:conc}
We introduced a labelled version of lqCCS~\cite{ceragioliQuantumBisimilarityBarbs2024} enriched by tag-based schedulers.
Resorting to simple schedulers allowed us to constrain processes so that they perform physically admissible choices only, i.e. independent of the quantum states.
This suffices for making our proposed saturated bisimilarity $\sim_s$ compliant with the limited observational power prescribed by quantum theory (\autoref{thm:propertyA}).
Moreover, $\sim_s$ is by definition a congruence with respect to parallel composition.
Finally, we characterized the atomic observable properties of lqCCS by deriving a labelled bisimilarity $\sim_l$, provably
equivalent to $\sim_s$ (\autoref{thm:corrcompl}).

To the best of our knowledge, our work is the first behavioural equivalence for distributed, non-deterministic quantum systems that
$(i)$ is a congruence for the parallel operator, $(ii)$ abides the prescriptions of quantum theory, and $(iii)$ makes explicit the observables through labels.

\paragraph{\textbf{Future Work.}}
A fourth desideratum of a behavioural equivalence is being decidable.
Our labelled bisimilarity goes in this direction, saving us from comparing processes under every possible context.
However, the prescribed closure for superoperators still requires considering an infinite number of cases.
As future work we will investigate if this condition can be safely removed from the definition of $\sim_l$, possibly in some specific cases.
This would allow the bisimilarity of distributions of configurations to be decided.
Moreover, we will investigate symbolic approaches for comparing lqCCS processes directly,
i.e. for guaranteeing that they are bisimilar for all “ground” systems obtained by instantiating the quantum input.
These approaches all maintain the same semantic model, that of probability distributions of configurations. 
We are investigating an alternative model, made of "quantum distributions", which generalize probability distributions by pairing a process with a single (partial) density operator, 
encoding both the quantum state and the probability, in the style of~\cite{concur2024}.
 Such semantics would satisfy \autoref{thm:propertyA} by construction, and seems more adequate to model quantum protocols.
Finally, a different line of research is to compare our tag-based approach with those of semantic schedulers, and to extend the capability of schedulers while preserving our results.

\bibliographystyle{splncs04}
\bibliography{references}

\clearpage
\appendix

\renewcommand{\com}[1]{}


\section{Full lqCCS}\label{sec:typeappendix}
The full type system for lqCCS
	{\small
		\begin{gather*}
			\begin{matrix}
				\infer[\rulename{Nil}]{\Sigma \vdash \nil_{\tilde{e}}}{\tilde{e} \in \tilde{\Sigma}} \qquad
				\infer[\rulename{Tau}]{\Sigma \vdash t \tags \tau . P}{\Sigma \vdash P} \qquad
				\infer[\rulename{TauPair}]{\Sigma \vdash (t, t') \tags \tau . P}{\Sigma \vdash P} \qquad
				\infer[\rulename{Restrict}]{\Sigma \vdash P \setminus c}{\Sigma \vdash P}                                                                                                                                                                                                                                                                                                                                                 \\[0.2cm]
				\infer[\rulename{Sum}]{\Sigma \vdash P + Q}{\Sigma \vdash P                                                & \Sigma \vdash Q}                                                                                                                                       \qquad
				\infer[\rulename{QOp}]{\Sigma \vdash t \tags \sop{E}{\tilde{e}} . P}{\mathcal{E} : \opset(n)               & |E| = n                                                                                                                                                       & \tilde{e} \in \tilde{E}                                                           & E \subseteq \Sigma & \Sigma \vdash P}                    \\[0.2cm]
				\infer[\rulename{QMeas}]{\Sigma \vdash t \tags \meas{\tilde{e}}{y} . P}{M : \measset(n)                    & |E| = n                                                                                                                                                       & \tilde{e} \in \tilde{E}                                                           & E \subseteq \Sigma & y : \ntype       & \Sigma \vdash P} \\[0.2cm]
				\infer[\rulename{CRecv}]{\Sigma \vdash t \tags c?x . P}{c : \chtype{T}                                     & x : T \in\{\btype,\ntype\}                                                                                                                                    & \Sigma \vdash P} \qquad
				\infer[\rulename{QRecv}]{\Sigma \vdash t \tags c?x . P}{c : \chtype{\qtype}                                & x : \qtype                                                                                                                                                    & \Sigma \cup \{x\} \vdash P}                                                                                                                  \\[0.2cm]
				\infer[\rulename{QSend}]{\Sigma \vdash t \tags c!e .P}{c : \chtype{\qtype}                                 & e \in \Sigma                                                                                                                                                  & \Sigma \setminus \{ e \} \vdash P}                                         \qquad
				\infer[\rulename{CSend}]{\Sigma \vdash t \tags c!e.P}{c : \chtype{T}                                       & e : T \in\{\btype,\ntype\}                                                                                                                                    & \Sigma \vdash P}                                                                                                                             \\[0.2cm]
				\infer[\rulename{ITE}]{\Sigma \vdash \ite{e}{P_1}{P_2}}{e : \btype                                         & \Sigma \vdash P_1                                                                                                                                             & \Sigma \vdash P_2} \qquad
				\infer[\rulename{Par}]{\Sigma_1 \cup \Sigma_2 \vdash P_1 \parallel P_2}{\Sigma_1 \cap \Sigma_2 = \emptyset & \Sigma_1 \vdash P_1                                                                                                                                           & \Sigma_2 \vdash P_2}
			\end{matrix}
		\end{gather*}}

\begin{lemma}\label{lem:quantumsubst}
	If $\Sigma \cup \{x\} \vdash P$ and $v \not\in \Sigma$ then $\Sigma \cup \{v\} \vdash P[\sfrac{v}{x}]$.
\end{lemma}
\begin{proof}
	By structural induction on $P$.
	Let $P = \nil_{\tilde{e}}$. If $\Sigma \cup \{x\} \vdash P$ then it must be that $x \not\in \Sigma$ and $x \in \tilde{e}$, thus trivially $\Sigma \cup \{v\} \vdash \nil_{\tilde{e}}[\sfrac{v}{x}]$.

	Let $P = c?y.P'$ for some process $P'$ where $c : \chtype{\qtype}$.
	The only applicable rule is \rulename{QRecv}.
	If $y = x$ then it trivially follows by inductive hypothesis.
	Otherwise, by inductive hypothesis $\Sigma \cup \{v, x\} \vdash P'$ and thus $\Sigma \cup \{v\} \vdash c?y.P$.

	Let $P = c!e.P'$ for some process $P'$ where $c : \chtype{\qtype}$.
	The only applicable rule is \rulename{QSend}.
	However, $e \neq v$ since $v \not\in \Sigma$ for any $\Sigma \vdash P$ by hypothesis.
	Then the conclusion follows trivially by inductive hypothesis.

	Let $P = P_1 \parallel P_2$ for some processes $P_1$ and $P_2$. By the hypothesis of the \rulename{Par} rule, $\Sigma_1 \vdash P_1$ and $\Sigma_2 \vdash P_2$ where $x$ is either in $\Sigma_1$ or $\Sigma_2$ but not in both.
	Without loss of generality, assume $x \in \Sigma_1$.
	By inductive hypothesis $\Sigma_1[\sfrac{v}{x}] \vdash P_1[\sfrac{v}{x}]$ and since $v \not\in \Sigma$ it is also true that $v \not\in \Sigma_2$, thus $\Sigma_2 \vdash P_2[\sfrac{v}{x}]$.
	Therefore, the mutual exclusivity requirement still holds, and we can reapply the \rulename{Par} rule.

	All the other cases are trivial application of the inductive hypothesis on the premises of the only applicable rule.
\qed\end{proof}

\quasitypingpreservation*
\begin{proof}
	Consider the function $f$ which takes as input a typing context and an action
	\[
		f(\Sigma, \mu) = \begin{cases}
			\Sigma \cup \{ q \}      & \text{if $\mu = c?q$ with $c : \chtype{\qtype}$} \\
			\Sigma \setminus \{ q \} & \text{if $\mu = c!q$ with $c : \chtype{\qtype}$} \\
			\Sigma                   & \text{otherwise}                                 \\
		\end{cases}
	\]
	Let $(\Sigma'', \Sigma) \vdash \iconf$ and $\iconf \xlongrightarrow{\mu}_s \Delta$. We will prove that $(\Sigma'', f(\Sigma, \mu)) \vdash \Delta$.

	By induction on the derivation of $\xlongrightarrow{\mu}_s$.
	The base cases $\rulename{Tau}$, $\rulename{TauPair}$, $\rulename{Send}$, $\rulename{Receive}$, $\rulename{QOp}$, and $\rulename{QMeas}$ 
	are trivial by correspondence with their respective typing rules.

	Cases $\rulename{ITEt}$, $\rulename{ITEf}$, $\rulename{SumL}$, $\rulename{SumR}$, and $\rulename{Restrict}$ hold by induction, since respectively by the typing rules $\rulename{Sum}$, 
	$\rulename{ITE}$ or $\rulename{Restrict}$, $\iconf$ and $\Delta$ must have the same type.

	For  case \rulename{ParL} (and symmetrically \rulename{ParR}), by the typing rule \rulename{Par} there exists $P_1$, $P_2$, $\Sigma_1$ and $\Sigma_2$ such that $\iconf = \conf{\rho, P_1 \parallel P_2}$, $\Sigma = \Sigma_1 \cup \Sigma_2$ with $\Sigma_1 \cap \Sigma_2 = \emptyset$ and $(\Sigma'', \Sigma_1) \vdash P_1$.
	By induction $(\Sigma'', f(\Sigma_1, \mu)) \vdash \Delta$. We must verify that $(\Sigma'', f(\Sigma_1, \mu) \cup \Sigma_2) \vdash \Delta \parallel P_2$ with $f(\Sigma_1, \mu) \cap \Sigma_2 = \emptyset$. If $\mu = \tau$ or $\mu = c!v, c?v$ with $c$ a classical channel, then $f(\Sigma_1, \tau) = \Sigma_1$, which trivially satisfies the requirements.
	If $c!v$ with $c : \chtype{\qtype}$ then $f(\Sigma_1, c!v) = \Sigma_1 \setminus \{v\}$, but by the \rulename{QSend} rule $v \in \Sigma_1$ and by \rulename{Par} $v \not\in \Sigma_2$, thus the the requirements hold.

	For case \rulename{SynchL} (and symmetrically \rulename{SynchR}), assume $\iconf = \conf{\rho, P \parallel Q}$ with $\conf{\rho, P} \xlongrightarrow{c!v}_t \sconf{\rho, P'}$, $\conf{\rho, Q} \xlongrightarrow{c?v}_{t'} \sconf{\rho, Q'}$ and $\Delta = \sconf{\rho, P' \parallel Q'}$.
	If $c$ is a classical channel then as for the parallel case, there are no change to the typing context as indicated by the update function $f$.
	If $c : \chtype{\qtype}$, then, as for the parallel case, there exists $\Sigma_1$ and $\Sigma_2$ such that $\Sigma = \Sigma_1 \cup \Sigma_2$ with $\Sigma_1 \cap \Sigma_2 = \emptyset$ and $(\Sigma'', \Sigma_1) \vdash P$, $(\Sigma'', \Sigma_2) \vdash Q$.
	By induction and application of $f$, $(\Sigma'', \Sigma_1 \setminus \{v\}) \vdash P'$ and $(\Sigma'', \Sigma_2 \cup \{v\}) \vdash Q'$. But, as for the previous case $v \in \Sigma_1$ thus $v \in \Sigma''$, and by \rulename{Par} $v \not\in \Sigma_2$.
	Thus, $f(\Sigma_1, c!v) \cup f(\Sigma_2, c?v) = (\Sigma_1 \setminus \{v\}) \cup (\Sigma_2 \cup \{v\}) = \Sigma = f(\Sigma, \tau)$.
\qed\end{proof}

\typingpreservation*
\begin{proof}
	Immediately follows from the proof~\autoref{thm:quasipreservation}, by noticing that for $\iconf \xlongrightarrow{\tau} \Delta$, the unique $\Sigma_\tau$ is exactly $\Sigma_P$.
\qed\end{proof}

Hereafter, we write $env(\Delta)$ for $\tr_{\Sigma'}(qs(\Delta))$ when $(\Sigma, \Sigma') \vdash \Delta$.

We now prove some results about the semantics and its lifting as LTS over distributions of configurations.
	Recall that $B[P] = P \parallel Q$ with $B[\blank] = [\blank] \parallel Q$.
\begin{lemma}\label{thm:contextpreservesmoves}
	For any $B[\blank]$ and $\Delta$, if $\Delta\! \tauarrow_\delta\! \Delta' \neq \overline{\bot}$ then $B[\Delta]\! \tauarrow_\delta\! B[\Delta'] \neq \overline{\bot}$.
\end{lemma}
\begin{proof}
	Note that $B[\Delta] = P \parallel \Delta$ for some $P$.
	Then, the move is obtained via the $\rulename{ParRight}$ rule thanks to decomposability and linearity of $\tauarrow_\delta$.
\qed\end{proof}

\begin{lemma}\label{thm:alwaysbot}
	For any $P$, then either one of the two conditions below holds
	\begin{itemize}
		\item for all $\rho$, $\sconf{\rho, P} \tauarrow_\delta \bot$, or
		\item for all $\rho$, $\sconf{\rho, P} \not\tauarrow_\delta \bot$.
	\end{itemize}
\end{lemma}
\begin{proof}
	By induction on the rules of the semantics.
\qed\end{proof}

\begin{lemma}\label{thm:bothbot}
	For any $P$, $t$ and $\Delta$, if $\Delta \tauarrow_{\singleton{t}} \bot$ and $\sconf{\rho, P} \tauarrow_{\singleton{t}} \bot$ for some $\rho$, then
	$\Delta \parallel P \tauarrow_{\singleton{t}} \bot$.
\end{lemma}
\begin{proof}
	By~\autoref{thm:alwaysbot}, decomposability and linearity of $\tauarrow_{\singleton{t}}$.
\qed\end{proof}

\begin{lemma}\label{thm:alwaysmove}
	For any density operator $\rho$, process $P$, and distribution of schedulers $\delta$, if $\sconf{\rho, P} \tauarrow_\delta \Delta \neq \bot$ then $\Delta = \sum_i p_i \conf{\frac{1}{tr(\E_i(\rho))}\E_i(\rho), P_i}$ for some $p_i$ and
	for some family of superoperators $\mathcal{E}_i$. 

	Moreover, for any $\sigma \neq \rho$, $\sconf{\sigma, P} \tauarrow_\delta \sum_i p_i' \conf{\frac{1}{tr(\E_i(\sigma))}\E_i(\sigma), P_i}$,
	and, if $tr_{\overline{\Sigma_P}}(\rho) = tr_{\overline{\Sigma_P}}(\sigma)$ then $p_i = p_i'$ for any $i$.
\end{lemma}
\begin{proof}
	By induction on the rules of the semantics.
\qed\end{proof}

\begin{lemma}\label{thm:simplemoves}
	For any $B[\blank]$, $\Delta$, $\Delta'$, and $t$, if $B[\Delta] \tauarrow_{\singleton{t}} \Delta'$ then one of the following conditions holds
	\begin{itemize}
		\item $\Delta' = B[\Delta'']$ for some $\Delta'' \neq \overline{\bot}$ such that $\Delta \tauarrow_{\singleton{t}} \Delta''$;
		\item $\Delta' = \overline{\bot}$ and $\Delta \tauarrow_{\singleton{t}} \overline{\bot}$ and $B[\Theta] \tauarrow_{\singleton{t}} \overline{\bot}$ for all $\Theta$ such that
		      $\Theta \tauarrow_{\singleton{t}}  \overline{\bot}$; or
		\item $\Delta' = \sum_i p_i B_i[\nE_i(\Delta)]$ for some $p_i$, $B_i[\blank]$,
		      and family of superoperators $\mathcal{E}_i$. 
		      Moreover, for any $\Theta$ such that $env(\Delta) = env(\Theta)$, $B[\Theta] \tauarrow_{\singleton{t}} \sum_i p_i B_i[\nE_i(\Theta)]$.
	\end{itemize}
\end{lemma}
\begin{proof}
	We consider the cases $\Delta \tauarrow_{\singleton{t}} \Delta'' \neq \overline{\bot}$ and its negation, which is equivalent by determinism to $\Delta \tauarrow_{\singleton{t}} \overline{\bot}$.
	In the first case, the transition is recovered thanks to~\autoref{thm:contextpreservesmoves}.
	Assume now that $\Delta \tauarrow_{\singleton{t}} \overline{\bot}$, and note that $B[\Delta] = P \parallel \Delta$ for some $P$.
	By~\autoref{thm:alwaysbot}, we can consider two cases: either $\conf{\rho, P} \tauarrow_{\singleton{t}} \overline{\bot}$ for all $\rho$ or for no $\rho$ at all.
	Moreover, by decomposability, $\confel \tauarrow_{\singleton{t}} \overline{\bot}$ for any $\confel \in supp(\Delta)$.
	Then the second case follows from~\autoref{thm:bothbot} and by linearity of $\tauarrow_{\singleton{t}}$.
	The last case follows similarly by~\autoref{thm:contextpreservesmoves} and~\autoref{thm:alwaysmove}, and by decompasbility and linearity of $\tauarrow_{\singleton{t}}$.
	These cases are exhaustive since our choice is $t$ and thus we have to consider only $\rulename{ParR}$ and $\rulename{ParL}$ moves.
	%
	%
\qed\end{proof}

\begin{lemma}\label{thm:syncmoves}
	For any $B[\blank]$, $\Delta$, $\Delta'$, $t_1$ and $t_2$, if $B[\Delta] \tauarrow_{\singleton{(t_1,t_2)}} \Delta'$ then one of the following conditions holds
	\begin{itemize}
		\item $\Delta' = B[\Delta'']$ for some $\Delta'' \neq \overline{\bot}$ such that $\Delta \tauarrow_{\singleton{(t_1,t_2)}} \Delta''$;
		\item $\Delta' = \overline{\bot}$ and $\Delta \tauarrow_{\singleton{(t_1,t_2)}} \overline{\bot}$ and $B[\Theta] \tauarrow_{\singleton{(t_1,t_2)}} \overline{\bot}$ for all $\Theta$ such that $\Theta \tauarrow_{\singleton{(t_1,t_2)}} \overline{\bot}$;
		\item $\Delta' = B'[\Delta]$ for some $B'[\blank]$, and for any $\Theta$, $B[\Theta] \tauarrow_{\overline{(t_1,t_2)}} B'[\Theta]$; or
		\item $\Delta' = B'[\Delta'']$ for some $\Delta''$ and $B'[\blank]$ such that $\Delta \xlongrightarrow{\mu}_{\singleton{t_i}} \Delta''$ with $i \in \{1,2\}$ and, for each $\Theta$ such that $\Theta \xlongrightarrow{\mu}_{\singleton{t_i}} \Theta'$, $B[\Theta] \tauarrow_{\singleton{(t_1,t_2)}} B'[\Theta']$.
	\end{itemize}
\end{lemma}
\begin{proof}
	The first three cases coincide with the ones of~\autoref{thm:simplemoves}, where in the third case we exploit the fact that schedulers with a pair of tags never update the quantum state.
	The last case coincides with moves derived using $\rulename{Synch}$, where the synchronization happens between a process in $\Delta$ and one in $B[\blank]$.
\qed\end{proof}


\section{Proofs of~\autoref{sec:ass}}


\begin{proposition}\label{thm:envtrace}
	$tr(env(\Delta)) = 1 - \Delta(\bot)$.
\end{proposition}
\begin{proof}
		Follows from the definition of $env$, and from the fact that $tr(\rho) = 1$ for any $\rho$ in a configuration $\conf{\rho, P}$.
\qed\end{proof}

We now introduce the notion of \emph{up-to bisimulation}~\cite{sangiorgienhancements2011}, that will be useful in our proofs.
	In particular, we will use bisimulations up-to context closure and up-to convex hull~\cite{bonchipower2017}.
\begin{definition}
	Let $\rel \subseteq \dist{\confbot} \times \dist{\confbot}$. The \emph{convex hull} $Cv(\rel)$ of $\rel$ is the least relation satisfying the following rule
	\[
		\infer{(\sum_{i \in I}\distelem{p_i}{\Delta_i})\,Cv(\rel)\,(\sum_{i \in I} \distelem{p_i}{\Theta_i})}{\forall i \in I \ldotp \Delta_i\,\rel\,\Theta_i}
	\]
	The \emph{context closure} $B(\rel)$ of $\rel$ is the least relation satisfying the following rule
	\[
		\infer{B[\Delta] \, B(\rel) \, B[\Theta]}
		{\Delta \rel \Theta}
	\]
	where $\Sigma \vdash \Delta, \Theta$ and $B[\blank]$ has a hole of type $\Sigma$.
\end{definition}
\begin{definition}[saturated bisimulation up-to]
	A relation $\rel \subseteq \dist{\confbot} \times \dist{\confbot}$ is a
	\emph{saturated bisimulation up-to $Cv \circ B$} if $\Delta\,\rel\,\Theta$ implies
	$(\Sigma, \Sigma') \vdash \Delta$ and $(\Sigma, \Sigma') \vdash \Theta$ for some $\Sigma, \Sigma'$, $\mass{\Delta} = \mass{\Theta}$ and for any context $B[\blank]$ it holds
	\begin{itemize}
		\item whenever $B[\Delta] \tauarrow_\delta \Delta'$, there exists $\Theta'$
		      such that $B[\Theta] \tauarrow_\delta \Theta'$ and $\Delta'\;Cv(B(\rel))\;\Theta'$;
		\item whenever $B[\Theta] \tauarrow_\delta \Theta'$, there exists $\Delta'$
		      such that $B[\Delta] \tauarrow_\delta \Theta'$ and $\Delta'\;Cv(B(\rel))\;\Theta'$.
	\end{itemize}
\end{definition}

We now prove that bisimulation up-to $Cv \circ B$ is \emph{valid}, i.e. that we can use it to prove bisimilarity.
We define the function $b$ over relations, of which saturated bisimilarity is the greatest fix point.
\begin{footnotesize}
	\[
		b(\rel) \coloneqq \left\{
		(\Delta, \Theta) \biggm| \begin{array}{c}
			\mass{\Delta} = \mass{\Theta} \\
			B[\Delta] \tauarrow_\delta \Delta' \Rightarrow \exists \Theta' \  B[\Theta] \tauarrow_\delta \Theta' \wedge \Delta'\,\rel\,\Theta'
			\\
			B[\Theta] \tauarrow_\delta \Theta' \Rightarrow \exists \Delta' \ B[\Delta] \tauarrow_\delta \Delta' \wedge \Delta'\,\rel\,\Theta'
		\end{array}
		\right\}
	\]
\end{footnotesize}

Observe that $b$, $Cv$ and $B$ are monotone functions on the lattice of relations.

\begin{lemma}[$Cv$ is $b$-compatible]\label{lem:cvcomp}
	We have that $\forall \rel\ldotp Cv(b(\rel)) \subseteq b(Cv(\rel))$.
\end{lemma}
\begin{proof}
	Assume $(\Delta, \Theta) \in Cv(b(\rel))$. Then it must be $\Delta = \sum_{i \in I}\distelem{p_i} \Delta_i$ and $\Theta = \sum_{i \in I}\distelem{p_i} \Theta_i$ for a certain set of probabilities $\{p_i\}_{i \in I}$, with
	$\Delta_i\,b(\rel)\,\Theta_i$. So for any $i \in I$ we have
	\begin{gather*}
		\Delta_i(\bot) = \Theta_i(\bot) = q_i \\
		\sum_{i \in I} p_i \Delta_i(\bot)  =
		\sum_{i \in I} p_i \Theta_i(\bot) = \sum_{i \in I} p_i q_i = q
	\end{gather*}
	meaning that $\mass{\Delta} = 1 - q = \mass{\Theta}$.

	Suppose $B[\Delta] = \sum_{i \in I} \distelem{p_i} B[\Delta_i]  \tauarrow_\delta \Delta'$.
	From~\cite{hennessyexploring2012}, $\tauarrow_\delta$ is left-decomposable, so it must be $\Delta'
		= \sum_{i \in I} \distelem{p_i} \Delta_i'$ with $B[\Delta_i] \tauarrow_\delta \Delta_i'$. But
	since $\Delta_i\,b(\rel)\,\Theta_i$ it must be $B[\Theta_i] \tauarrow_\delta
		\Theta_i'$ with $\Delta_i'\,\rel\,\Theta_i'$ for any $i \in I$, from which it follows that
	$\sum_{i \in I} \distelem{p_i} B[\Theta_i] \tauarrow_\delta \sum_{i \in I} \distelem{p_i} \Theta_i' = \Theta'$.

	In other words, $\Delta$ and $\Theta$ express the same barb and whenever
	$B[\Delta] \tauarrow_\delta \Delta'$, there exists a transition $B[\Theta]
		\tauarrow_\delta \Theta'$ such that $\Delta'\,Cv(\rel)\,\Theta'$ (the symmetrical
	argument is the same). So we can conclude that $(\Delta, \Theta) \in
		b(Cv(\rel))$.
\qed\end{proof}

\begin{lemma}[$B$ is $b$-compatible]\label{lem:bcomp}
	We have that $\forall \rel\ldotp B(b(\rel)) \subseteq b(B(\rel))$.
\end{lemma}
\begin{proof}
	Assume $\Delta\ b(\rel)\ \Theta$, we will check that 
	$B[\Delta]$ and $B[\Theta]$
	are in $b(B(\rel))$. The case for $B[\blank] = \blank$ is trivial.
	For the first condition, it is easy to see that $B[\Delta](\bot) = \Delta(\bot)$, because the context $B[\blank]$ is applied linearly and $B[\bot] = \bot$.
	Thus, we have that $\mass{B[\Delta]} = \mass{\Delta} = \mass{\Theta} = \mass{B[\Theta]}$.
	For the second one, we can show that the parallel operator is associative, $(P \parallel Q) \parallel R) \sim_s P \parallel (Q \parallel R)$. Therefore, the desired condition on $B'[B[\Delta]] = (\Delta \parallel R) \parallel R'$ follows from the fact that $B'[B[\Delta]] \sim_s B''[\Delta] = \Delta \parallel (R \parallel R')$ and the fact that $\Delta\ b(\rel)\ \Theta$.
\qed\end{proof}

\begin{restatable}{theorem}{uptoValid}\label{thm:uptovalid}
	Bisimulation up-to $Cv \circ B$ is a valid proof technique for $\sim_S$, meaning that if $\Delta\ \rel\ \Theta$ for a saturated bisimulation up-to 
	$Cv \circ B$,
	 then $\Delta \sim_s \Theta$.
\end{restatable}
\begin{proof}
	From~\cite{sangiorgienhancements2011}, we know that when two function $f_1, f_2$ are b-compatible, then also $f_1 \circ f_2$ is b-compatible.  Furthermore, bisimulations up-to $f$ are a sound proof technique whenever $f$ is compatible.
\qed\end{proof}

Besides, this also allows us to show as a corollary that $\sim_{s}$ is linear.

\linearityCongruence*
\begin{proof}
	From \autoref{lem:cvcomp} we know that $Cv$ is a b-compatible function.
	We now prove that $f(\sim_{cs}) \subseteq\ \sim_{cs}$ for any compatible $f$.
	Since $f$ is $b$-compatible, we have that $f(b(\sim_{cs})) \subseteq b(f(\sim_{cs}))$, and
	since $\sim_{cs}$ is the greatest fix point of $b$, we have $f(\sim_{cs}) \subseteq
		b(f(\sim_{cs}))$, meaning that $f(\sim_{cs})$ is a bisimulation, and so
	$f(\sim_{cs}) \subseteq\ \sim_{cs}$.
\qed\end{proof}

\begin{lemma}\label{thm:botDecomposability}
	For any $p$, $\Delta \psum{p} \bot \sim_s \Theta \psum{p} \bot$ if and only if $\Delta \sim_s \Theta$.
\end{lemma}
\begin{proof}
	$\Delta \sim_s \Theta$ implies $\Delta \psum{p} \bot \sim_s \Theta \psum{p} \bot$ by linearity. 
	For the other side, it suffices noticing that $\rel = \{\Delta, \Theta  \mid \Delta \psum{p} \bot \sim_s \Theta \psum{p} \bot\}$ is a bisimulation, thanks to the facts that $\mass{\Delta \psum{p} \bot} = p\mass{\Delta} = p\mass{\Theta} = \mass{\Theta \psum{p} \bot}$, and $B[\Delta \psum{p} \bot] = B[\Delta] \psum{p} \bot$.
\qed\end{proof}


\propertyA*
\begin{proof}
	We will prove that $\singleton{\conf{\rho \psum{p} \sigma, P}} \,\sim_s\, \singleton{\conf{\rho, P}} \psum{p} \singleton{\conf{\sigma, P}}$ for any operators $\rho, \sigma$, process $P$ and probability $p$.
	This is sufficient to prove out theorem, as thanks to transitivity and linearity(\autoref{thm:linearity}) of $\sim_s$ we have
	\[ \sum_i \distelem{p_i}{\sconf{\rho_i, P}} \sim_s^* \sconf{\sum_i p_i\cdot \rho_i, P} = \sconf{\sum_j q_j\cdot \sigma_j, P} \sim_s^* \sum_j \distelem{q_j}{\sconf{\sigma_j, P}}.
	\]
	We define
	\[
		\rel = \{(\bot, \bot)\} \cup \left\{ \left(\singleton{\conf{\rho \psum{p} \sigma, P}} \ , \ \singleton{\conf{\rho, P}} \psum{p} \singleton{\conf{\sigma, P}}\right) \mid \rho, \sigma, p, P \right\}
	\]
	and prove that it is a bisimulation up-to $Cv$. That is, we require that if $B[\Delta] \tauarrow_\delta \Delta'$ then $B[\Theta] \tauarrow_\delta \Theta'$ with $\Delta' Cv(\rel) \Theta'$.

	The case $\Delta = \Theta = \bot$ is straightforward. Otherwise, let $\Delta = \singleton{\conf{\nu, P}}, \Theta = \singleton{\conf{\rho, P}} \psum{p} \singleton{\conf{\sigma, P}}$ be in $\rel$, with $\nu = \rho \psum{p} \sigma$. Since they have the same process, they are typed by the same context $\Sigma$. Notice that we do not need to quantify over any context $B[\blank]$, because  $\rel$ is a saturated relation, meaning that if $\Delta, \Theta \in \rel$, then $B[\Delta], B[\Theta] \in \rel$ for any context.

	For the first condition, we have that $\Delta(\bot) = \Theta(\bot) = 0$.

	For the second condition, suppose that $\singleton{\conf{\nu, P}} \tauarrow_\delta \Delta'$.
	When $\Delta' = \bot$, then also $\Theta \tauarrow_\delta \bot$ by~\autoref{thm:alwaysbot} and by the linearity of $\tauarrow$.
	Otherwise, we will proceed by induction on $\tauarrow_\delta$ to prove that $\Theta \tauarrow_\delta \Theta'$ with $\Delta' Cv(\rel) \Theta'$. The ``classical'' cases, which do not modify the quantum state, are trivial, since if $\Delta \tauarrow_\delta \Delta'$ then $\Theta$ can go in $\Theta'$ performing the choice $\delta$ in both configurations, and $\Delta' \rel \Theta'$. The only interesting cases are $\textsc{QOp}$ and $\textsc{QMeas}$.

	In the $\rulename{QOp}$ case,
	\begin{gather*}
		\text{if} \quad
		\singleton{\conf{\nu, t \tags\sop{E}{\tilde{x}}.P'}} \tauarrow_{\singleton{t}} \singleton{\conf{\mathcal{E}^{\tilde{x}}(\nu), P'}} \\
		\text{then} \quad
		\singleton{\conf{\rho, t \tags\sop{E}{\tilde{x}}.P'}} \psum{p} \singleton{\conf{\sigma, t \tags\sop{E}{\tilde{x}}.P'}}\ \tauarrow_{{\singleton{t}}}\ \singleton{\conf{\mathcal{E}^{\tilde{x}}(\rho), P'}} \psum{p} \singleton{\conf{\mathcal{E}^{\tilde{x}}(\sigma), P'}}
	\end{gather*}
	and $\mathcal{E}^{\tilde{x}}(\rho) \psum{p} \mathcal{E}^{\tilde{x}}(\sigma) = \mathcal{E}^{\tilde{x}}(\rho \psum{p} \sigma) = \mathcal{E}^{\tilde{x}}(\nu)$, thanks to linearity of superoperators.

	In the $\rulename{QMeas}$ case, we have
	\[ \singleton{\conf{\nu, t \tags\meas{\tilde{x}}{y}.P'}} \tauarrow_{\singleton{t}}  \Delta'
		\quad
		\text{with }
		\Delta' = \sum_m \distelem{tr_m(\nu)}\singleton{\conf{\nu'_m, P'[\sfrac{m}{y}]}},
		\quad \nu'_m = \frac{\mathcal{M}_m^{\tilde{x}}(\nu)}{tr_m(\nu)}
	\]
	where $tr_m{\nu} = tr(\mathcal{M}_m^{\tilde{x}}(\nu))$ is the probability of said outcome.
	Then
	\[
		\singleton{\conf{\rho, t \tags\meas{\tilde{x}}{y}.P'}} \psum{p} \singleton{\conf{\sigma, t \tags\meas{\tilde{x}}{y}.P'}}
		\tauarrow_{\singleton{t}} \Theta'_\rho \psum{p} \Theta'_\sigma
	\]
	with
	\begin{align*}
		\Theta'_\rho   & = \sum_m \distelem{tr_m(\rho)} \singleton{\conf{\rho'_m, P'[\sfrac{m}{y}]}} \qquad \rho'_m = \frac{\mathcal{M}_m^{\tilde{x}}(\rho)}{tr_m(\rho)}
		\\
		\Theta'_\sigma & = \sum_m \distelem{tr_m(\sigma)} \singleton{\conf{\sigma'_m, P'[\sfrac{m}{y}]}} \qquad \sigma'_m = \frac{\mathcal{M}_m^{\tilde{x}}(\sigma)}{tr_m(\sigma)}
	\end{align*}
	Observe that $tr_m(\nu) = tr(\mathcal{M}_m^{\tilde{x}}(\rho \psum{p} \sigma))$ is equal to $tr_m(\rho) \psum{p} tr_m(\sigma)$, thanks to linearity of superoperators and trace. So, according to the rules of probability distributions, $\Theta'_\rho \psum{p} \Theta'_\sigma$ can be rewritten as
	\[
		\sum_m \distelem{tr_m(\nu)} (\singleton{\conf{\rho'_m, P'[\sfrac{m}{y}]}} \psum{q} \singleton{\conf{\sigma'_m, P'[\sfrac{m}{y}]}})
	\]
	with $q = \frac{p\cdot tr_m(\rho)}{tr_m(\rho) \psum{p} tr_m{\sigma}}$. It is easy to show that $\rho'_m \psum{q} \sigma'_m = \nu'_m$, from which it follows that
	\[
		\singleton{\conf{\nu'_m, P'[\sfrac{m}{y}]}} \ \rel \ \left(\singleton{\conf{\rho'_m, P'[\sfrac{m}{y}]}} \psum{q} \singleton{\conf{\sigma'_m, P'[\sfrac{m}{y}]}}\right)
	\]
	and $\Delta' \ Cv(\rel)  \ \left(\Theta'_\rho \psum{p} \Theta'_\sigma\right)$.

	For the third condition, observe that thanks to the deterministic tagging there is exactly one transition $\conf{\rho, P} \xlongrightarrow{\tau}_\delta$ going out of $\conf{\rho, P}$ under the scheduler $\delta$. Furthermore, thanks to \autoref{thm:alwaysmove}, that same transition must be replicated also by $\conf{\sigma, P}$ and $\conf{\nu, P}$ under the same scheduler. For this reason we can prove the desired condition, $\Theta \tauarrow_\delta \Theta'$ implies $\Delta \tauarrow_\delta \Delta'$ for some $\Delta' Cv(\rel) \Theta'$, by proceeding by induction on $\conf{\rho, P} \xlongrightarrow{\tau}_\delta$, as the transitions of $\conf{\rho, P}$ are in bijection with the transitions $\Theta \tauarrow_\delta \Theta'$.
	The proof by induction is thus identical to the previous case.

\qed\end{proof}

%

\section{Proofs of~\autoref{sec:siml}}
\begin{lemma}\label{thm:alphamove}
For any $\Delta$ and $t$, and for any tag $t'$ fresh in $\Delta$, 
\[
\Delta \tauarrow_\delta \Delta' \quad \text{ if and only if }\quad \Delta[\sfrac{t'}{t}] \tauarrow_{\delta[\sfrac{t'}{t}]} \Delta'[\sfrac{t'}{t}]
\]
\end{lemma}
\begin{proof}
By induction on the operational semantics, then the case of $\bot$ and the lifting is trivial.
\qed\end{proof}

\begin{lemma}\label{thm:alpha}
For any $\Delta$, $\Theta$ and $t$, and for any tag $t'$ fresh in $\Delta$ and $\Theta$, 
\[
\Delta \sim_s \Theta \quad\text{ if and only if }\quad \Delta[\sfrac{t'}{t}] \sim_s \Theta[\sfrac{t'}{t}]
\]
\end{lemma}
\begin{proof}
Let $\rel$ be the smallest relation such that $\sim_s \subseteq \rel$ and $\rel \subseteq \{ (\Delta[\sfrac{t'}{t}], \Theta[\sfrac{t'}{t}]) \mid  \Delta \rel \Theta \text{ and } t' \text{ is fresh in }\Delta, \Theta\}$. 
We will show that $\rel$ is a bisimulation.

In the following, we write $[\sfrac{t_i'}{t_i}]_{i \in I}$ for the substitution $[\sfrac{t_{x_0}'}{t_{x_0}}][\sfrac{t_{x_1}'}{t_{x_1}}] \dots  [\sfrac{t_{x_n}'}{t_{x_n}}]$, with $I = \{x_0, x_1, \dots, x_n\}$.

Take $\Delta[\sfrac{t_i'}{t_i}]_{i \in I}$ and $\Theta[\sfrac{t_i'}{t_i}]_{i \in I}$ in $\rel$, with $\Delta \sim_s \Theta$, and take a generic context $B[\blank] = [\blank] \parallel R$.
Assume that $\Delta[\sfrac{t_i'}{t_i}]_{i \in I} \parallel R \tauarrow_\delta \Delta'$.

Take then a collection of distinct fresh tags $t_i''$.
By~\autoref{thm:alphamove},
\[
(\Delta[\sfrac{t_i'}{t_i}]_{i \in I} \parallel R)
[\sfrac{t_i''}{t_i},\sfrac{t_i}{t_i'}]_{i \in I}
 =
\Delta \parallel R[\sfrac{t_i''}{t_i},\sfrac{t_i}{t_i'}]_{i \in I} \tauarrow_{\delta[\sfrac{t_i''}{t_i},\sfrac{t_i}{t_i'}]_{i \in I}} \Delta'[\sfrac{t_i''}{t_i},\sfrac{t_i}{t_i'}]_{i \in I}.
\]
Since $\Delta \sim_s \Theta$, then it must be that $\Theta \parallel R[\sfrac{t_i''}{t_i},\sfrac{t_i}{t_i'}]_{i \in I} \tauarrow_{\delta[\sfrac{t_i''}{t_i},\sfrac{t_i}{t_i'}]_{i \in I}} \Theta'$, with $\Delta'[\sfrac{t_i''}{t_i},\sfrac{t_i}{t_i'}]_{i \in I} \sim_s \Theta'$.
By~\autoref{thm:alphamove}, $(\Theta \parallel R[\sfrac{t_i''}{t_i},\sfrac{t_i}{t_i'}]_{i \in I})[\sfrac{t_i'}{t_i},\sfrac{t_i}{t_i''}]_{i \in I} = \Theta  [\sfrac{t_i'}{t_i}]_{i \in I} \parallel R \tauarrow_{\delta} \Theta'[\sfrac{t_i'}{t_i},\sfrac{t_i}{t_i''}]_{i \in I}$.
We conclude by noting that $\Delta'\ \rel\ \Theta'[\sfrac{t_i'}{t_i},\sfrac{t_i}{t_i''}]_{i \in I}$ because $\Delta' = \Delta'[\sfrac{t_i''}{t_i},\sfrac{t_i}{t_i'}]_{i \in I}[\sfrac{t_i'}{t_i},\sfrac{t_i}{t_i''}]_{i \in I}$.
\qed\end{proof}

\begin{lemma}\label{thm:freshtag}
For any $\Delta$, $\Delta'$, and $B[\blank]$, and for any $\delta$ whose support contains only tags that are fresh in $B[\blank]$,
\[
B[\Delta] \tauarrow_\delta \Delta' \quad \text{ if and only if }\quad \Delta \tauarrow_\delta \Delta'' \text{ and } \Delta' = B[\Delta'']
\]
\end{lemma}
\begin{proof}
By the rules for the parallel composition.
\qed\end{proof}

\begin{lemma}\label{thm:ignoreSending}
	For any $\Sigma, \Sigma' \vdash \Delta, \Theta$, for any $q$ in $\Sigma \setminus \Sigma'$, we have that 
	$\Delta \sim_s \Theta$ if and only if
	$\Delta \parallel t\tags c!q \sim_s \Theta \parallel t\tags c!q$ for some $t$ fresh both in $\Delta$ and $\Theta$.
\end{lemma}
\begin{proof}
	$\Delta \sim_s \Theta$ implies $\Delta \parallel t\tags c!q \sim_s \Theta \parallel t\tags c!q$ because saturated bisimilarity is a congruence by definition. 
	For the other side, it suffices showing that $\rel = \{(\Delta, \Theta)  \mid \Delta \parallel t\tags c!\tilde{q} \sim_s \Theta \parallel t\tags c!\tilde{q}\}$ is a bisimulation.
	
	Take $(\Delta,\Theta) \in \rel$ and $\Sigma'' \vdash B[\blank]_{\Sigma'} = [\blank] \parallel R$.
	We consider two cases depending on $q$ being in $\Sigma''$ or not.
	
	If $q \in \Sigma''$, then let $\Sigma'' \vdash B'[\blank]_{\Sigma' \cup \{q\}} = [\blank] \parallel t'\tags c?x.R[\sfrac{x}{q}][\sfrac{t''}{t}]$, with $t'$ and $t''$ fresh tags in $\Delta, \Theta$ and $R$.
	We have then that $B'[\Delta \parallel t\tags c!q] \tauarrow_{{\singleton{(t,t')}}} \Delta \parallel R[\sfrac{t''}{t}]$, and 
	$B'[\Theta \parallel t\tags c!q] \tauarrow_{{\singleton{(t,t')}}} \Theta \parallel R[\sfrac{t''}{t}]$.
	Since $\Delta \parallel t\tags c!q \sim_s \Theta \parallel t\tags c!q$, we know that $\Delta \parallel R[\sfrac{t''}{t}] \sim_s \Theta \parallel R[\sfrac{t''}{t}]$, and the property follows by~\autoref{thm:alpha}.
	
	If $q \notin \Sigma''$, then let $\Sigma'' \vdash B'[\blank]_{\Sigma'} = [\blank] \parallel R[\sfrac{t'}{t}]$, with $t' \neq t$ fresh in $\Delta, \Theta$ and $R$.
	Then, assume $B[\Delta] \tauarrow_{\delta} \Delta'$. 
	Notice that $B[\Delta][\sfrac{t'}{t}] = B'[\Delta]$ since $t$ is fresh in $\Delta$.
	Hence, $B'[\Delta] \tauarrow_{\delta'} \Delta'[\sfrac{t'}{t}]$, with $\delta' = \delta[\sfrac{t'}{t}]$, by~\autoref{thm:alphamove}.
	Clearly, $t$ does not appear in $\delta$, hence, by~\autoref{thm:freshtag}, $B'[\Delta \parallel t\tags c!q] \tauarrow_{\delta'} \Delta'[\sfrac{t'}{t}] \parallel t\tags c!q$ (where we exploit the associativity of the parallel operator).
	Since $\Delta \parallel t\tags c!q \sim_s \Theta \parallel t\tags c!q$, we know that
	$B'[\Theta \parallel t\tags c!q] \tauarrow_{\delta'} \Theta'$, with $\Delta'[\sfrac{t'}{t}] \parallel t\tags c!q \sim_s \Theta'$.
	Moreover, by~\autoref{thm:freshtag}, $\Theta'$ is of the form $\Theta'' \parallel t\tags c!q$, and $B'[\Theta] \tauarrow_{\delta'} \Theta''$.
	By~\autoref{thm:alphamove}, we then know that $\Theta'' = \Theta'''[\sfrac{t'}{t}]$ for some $\Theta$ such that $B[\Theta] \tauarrow_{\delta} \Theta'''$.
	We know that $\Delta'[\sfrac{t'}{t}] \parallel t\tags c!q \sim_s \Theta'''[\sfrac{t'}{t}] \parallel t\tags c!q$, hence, by~\autoref{thm:alphamove}, 
	\begin{align*}
	(\Delta'[\sfrac{t'}{t}] \parallel t\tags c!q)[\sfrac{t}{t'}, \sfrac{t'}{t}] = \Delta' \parallel t'\tags c!q \sim_s 
	\Theta''' \parallel t\tags c!q = 
	(\Theta'''[\sfrac{t'}{t}] \parallel t\tags c!q)[\sfrac{t}{t'}, \sfrac{t'}{t}].
	\end{align*}
	Therefore, $(\Delta', \Theta''') \in \rel$.
\qed\end{proof}

In the following we will write $\tilde{t}\tags c!\tilde{q}$ for the process 
\[
t_0\tags c!q_0 \parallel t_1\tags c!q_1 \parallel \dots \parallel t_n\tags  c!q_n,
\]
with $c$ any channel, and $\tilde{t} = t_0, t_1, \dots t_n$ and $\tilde{q} = q_0, q_1, \dots q_n$ any sequences of (distinct) tags and qubit names.

\begin{lemma}\label{lem:cinesi}
	Let $\Delta \sim_s \Theta$. Then
	\begin{itemize}
		\item $env(\Delta) = env(\Theta)$;
			\item $\nE(\Delta) \sim_s \nE(\Theta)$ for any trace non-increasing superoperator over the environment.
	\end{itemize}
\end{lemma}
\begin{proof}
	For the first point, suppose that $env(\Delta) \neq env(\Theta)$, we will show that $\Delta \not\sim_s \Theta$.
		      Since $env(\Delta)  \neq env(\Theta)$, then there exists a positive linear operator $E$ such that $tr(E \cdot env(\Delta)) \neq tr(E \cdot env(\Theta))$~\cite{heinosaariMathematicalLanguageQuantum2011}.
		      We can construct the measurement $\mathbb{M}_{E}$ with measurement operators $\{\sqrt{E}, \sqrt{\mathbb{I} - E}\}$ and the context
		      $B[\blank] = [\blank] \parallel t_0\tags{\mathbb{M}_{E}}(\tilde{q} \rhd x).R$, where $R = \ite{x = 0}{t_1\tags\tau.\nil_{\tilde{q}}}{\nil_{\tilde{q}}}$, $t_0, t_1$ are fresh tags and $\tilde{q}$ is the tuple of all the qubits outside of $\Delta$.
		      For $\Delta = \sum_i \distelem{p_i}{\sconf{\rho_i, P_i}}$ and $\Theta = \sum_j \distelem{p_j}{\sconf{\rho_j, P_i}}$ we have
		      \begin{align*}
			      B[\Delta]
			       & \tauarrow_{\singleton{t_0}} \sum_i \distelem{p_i}{\left(\sconf{\rho_i^E, P_i \parallel R[0/x]} \psum{q_i} \sconf{\rho_i^{\not E}, P_i \parallel R[1/x]}\right)}  \\
			       & \tauarrow_{\singleton{t_1}} \left(\sum_i \distelem{\frac{p_iq_i}{r_\Delta}}{\sconf{\rho_i^E, P_i \parallel \nil_{\tilde{q}}}}\right) \psum{r_\Delta} \bot
		      \end{align*}
		      and
		      \begin{align*}
			      B[\Theta]
			       & \tauarrow_{\singleton{t_0}} \sum_j \distelem{p_j}{\left(\sconf{\rho_j^E, P_j \parallel R[0/x]} \psum{q_j} \sconf{\rho_j^{\not E}, P_j \parallel R[0/x]}\right)}  \\
			       & \tauarrow_{\singleton{t_1}} \left(\sum_j \distelem{\frac{p_jq_j}{r_\Theta}}{\sconf{\rho_j^E, P_j \parallel \nil_{\tilde{q}}}}\right) \psum{r_\Theta} \bot
		      \end{align*}
		      where $\rho_n^E$ (respectively $\rho_n^{\not E}$) is the state resulting from the successful (failed) measurement of $E$ on $\rho_n$, $q_n$ is the probability $tr(E\otimes \mathbb{I} \cdot \rho_n)$, and $r_\Delta$ (respectively $r_\Theta$) is equal to $\sum_i p_iq_i$ ($\sum_jp_jq_j$).
		      We have that $tr(E \cdot env(\Delta)) = tr(E\otimes I \cdot \sum_ip_i\rho_i) = r_\Delta$, and thus $r_\Delta \neq r_\Theta$ and $\Delta, \Theta$ are not bisimilar.
		      
		 For the second point, we consider separately the case of a trace-preserving superoperator $\mathcal{E}$.
		 In this case, we take the distinct fresh tags $t$ and $\tilde{t}$, and we build a context $B[\blank] = [\blank] \parallel t\tags\E(\tilde{q}). \tilde{t}\tags c!\tilde{q}$.
		After a $\xlongrightarrow{\tau}_{t}$ transition, $\Delta$ goes in $\nE(\Delta) \parallel \tilde{t}\tags c!\tilde{q}$ and $\Theta$ goes in $\nE(\Theta) \parallel \tilde{t}\tags c!\tilde{q}$, from which we know $\nE(\Delta) \sim_s \nE(\Theta)$ thank to \autoref{thm:ignoreSending}.
		
		Finally, 
		Suppose $\E = \{K_0, \ldots, K_{n-1}\}$ is a trace non-increasing superoperator with $n$ Kraus operators, we define $n$ superoperators $\mathcal{M}_i = \{K_i\}$ for $i = 0,\ldots, n-1$.
		We have that, for any $\rho$, $\E(\rho) = \sum_{i = 0}^{n-1} \mathcal{M}_i(\rho)$ and $tr(\E(\rho)) = \sum_{i = 0}^{n-1} tr(\mathcal{M}_i(\rho))$.
		Since $\E$ is trace non-increasing, we know by definition that $\I - \sum_{i = 0}^{n-1} K_i^\dagger K_i$ is a positive matrix, and we call this difference $M$.
		We build a measurement $\mathbb{M}$ with measurement operators $\{K_0, \ldots, K_{n-1}, \sqrt{M}\}$.
		This is  a $n+1$-outcome measurement, one for each Kraus operator of $\E$ and an additional one signifying that $\E$ does not happen.
		Take then the context $B[\blank] = [\blank] \parallel t\tags\mathbb{M}(\tilde{q} \rhd x).R$ where $R=\ite{x \neq n}{t'\tags\tau.\tilde{t}\tags c!\tilde{q}}{\tilde{t}\tags c!\tilde{q}}$ and $t, t', \tilde{t}$ are distinct fresh tags.
		Suppose that $\Delta = \sum_j \distelem{p_j}{\conf{\rho_j, P_j}}$, then a possible sequence of transition of $B[\Delta]$ is
		\begin{align*}
			&B[\Delta] \tauarrow_{\overline{t}} \sum_j \distelem{p_j}{\left(
				\sum_{i = 0}^{n} \distelem{tr(\mathcal{M}_i(\rho_j))}{\slrconf{\frac{\mathcal{M}_i(\rho_j)}{tr(\mathcal{M}_i(\rho_j))}, P_j \parallel R[i/x]}}
				\right)} \\
			& \tauarrow_{\overline{t'}}\sum_j \distelem{p_j }{\left[ 
				\left( 
				\sum_{i=0}^{n-1} \distelem{\frac{tr(\mathcal{M}_i(\rho_j))}{tr(\E(\rho_j))}}{\slrconf{\frac{\mathcal{M}_i(\rho_j)}{tr(\mathcal{M}_i(\rho_j))}, P_j \parallel \tilde{t}\tags c!\tilde{q} }}
				\right) \psum{tr(\E(\rho_j))}\bot
				\right]} = \Delta'
		\end{align*}
		Applying \autoref{thm:propertyA} to such $\Delta'$, we get by linearity
		\begin{align*}
			 \Delta' &\sim \sum_j \distelem{p_j}{\left[
				\overline{\left\langle{\frac{\sum_{i=0}^{n-1} \mathcal{M}_i(\rho_j)}{tr(\E(\rho_j))}, P_j \parallel \tilde{t}\tags c!\tilde{q} }\right\rangle}
				\psum{tr(\E(\rho_j))} \bot 
			\right]} \\
			&= \left(\sum_j \distelem{\frac{p_j tr(\E(\rho_j))}{p_{\E}}}{\overline{\left\langle\frac{\E(\rho_j)}{tr(\E(\rho_j))}, P_j \parallel \tilde{t}\tags c!\tilde{q}\right\rangle}}\right) \psum{p_{\E}} \bot \\
			&= (\nE(\Delta) \parallel \tilde{t}\tags c!\tilde{q}) \psum{p_{\E}} \bot
		\end{align*}
		where $p_{\E} = \sum_j p_j tr(\E(\rho_j))$ is the total probability of observing $\E$.
		This sequence of transition must be replicated also by $B[\Theta]$, and so we know that 
		\begin{align*}
			B[\Delta] \tauarrow_{t}\tauarrow_{t'} \ &\Delta' \sim (\nE(\Delta) \parallel \tilde{t}\tags c!\tilde{q}) \psum{p_{\E}} \bot \\
			&\ \rotatebox{90}{$\sim$} \\
			 B[\Theta] \tauarrow_{t}\tauarrow_{t'} \ &\Theta' \sim (\nE(\Theta) \parallel \tilde{t}\tags c!\tilde{q}) \psum{q_{\E}} \bot
		\end{align*}where $q_{\E}$ is the probability of observing $\E$ in $\Theta$.
		We thus know that $p_{\E} = q_{\E}$ (as it is expected, since $env(\Delta) = env(\Theta)$) and that $\nE(\Delta) \sim \nE(\Theta)$ thanks to \autoref{thm:botDecomposability} and \autoref{thm:ignoreSending}.
		
\qed\end{proof}

\begin{restatable}{theorem}{complete}\label{thm:complete}
	$\sim_s \ \subseteq \ \sim_l$.
\end{restatable}
\begin{proof}
	We have to prove that $\sim_s$ is a labelled bisimulation. The first two conditions are already proven in \autoref{lem:cinesi}, we must check the last two.
	Suppose $\Delta \sim_s \Theta$ and $\Delta \xlongrightarrow{\mu}_\delta \Delta'$. We proceed by cases on $\mu$.

	If $\mu = \tau$, then $B[\Delta] \tauarrow_\delta \Delta'$ for the empty context $B[\blank] = \blank$, thus there exists a $\Theta'$ such that $B[\Theta] = \Theta \tauarrow_\delta \Theta'$ and $\Delta' \sim_s \Theta'$.

	If $\mu = c?v$, then $\Delta \parallel t \tags c!v \tauarrow_{\delta \parallel t} \Delta' \parallel \nil$, where given a randomized scheduler $\delta$, $\delta \parallel t$ is defined linearly by $\singleton{t'} \parallel t = \singleton{(t', t)}$ and $(\delta_1 \psum{p} \delta_2) \parallel t = (\delta_1 \parallel t) \psum{p} (\delta_2 \parallel t)$. Note that $\delta \parallel t$ is not defined when $\delta = \singleton{h}$ or $\singleton{(t_1, t_2)}$, but we do not need these cases as we know that $\Delta$ performs a visible action $c?v$ under the scheduler $\delta$.
	Since $\Delta \parallel t \tags c!v \tauarrow_{\delta\parallel t} \Delta' \parallel \nil$, then there exists a $\Theta'$ such that $\Theta \parallel t\tags c!v \tauarrow_{\delta\parallel t} \Theta' \parallel \nil$. From the semantics of synchronization we also know that $\Theta \xlongrightarrow{c?v}_\delta \Theta'$ and $\Delta' \parallel \nil \sim_s \Theta' \parallel \nil$.
	It is possible to show that $\Delta \parallel \nil \sim_s \Delta$ for any $\Delta$, and thus we get $\Delta' \sim_s \Theta'$.

	If $\mu = c!v$, when $v$ is a quantum name, we build the context $B[\blank] = \blank \parallel t \tags c?x.R$ with $R = \ite{x = v}{t'\tags \tau. t''\tags c!x}{t''\tags c!x}$. Then we have that $B[\Delta] \tauarrow_{\delta \parallel t} \Delta' \parallel R[v/x]$, and thus $B[\Theta] \tauarrow_{\delta \parallel t} \Theta' \parallel R[u/x]$ and $\Theta \xlongrightarrow{c!u}_\delta \Theta'$ for some quantum name $u$ possibly different from $v$.
	Thanks to the bisimilarity between $\Delta' \parallel R[v/x]$ and $ \Theta' \parallel R[u/x]$, we can prove $u = v$ : $\Delta' \parallel R[v/x] \tauarrow_{t'} \Delta' \parallel t''\tags c!v$, and also $\Theta' \parallel  R[u/x]$ must have the same action available, so it must be $u = v$.
	To sum up, we know that $\Theta \xlongrightarrow{c?v} \Theta'$ and that $\Delta' \parallel t''\tags c!v \sim_s \Theta' \parallel t''\tags c!v$.
	This implies $\Delta' \sim_s \Theta'$ thanks to \autoref{thm:ignoreSending}.
	The case for $\mu = c!v$ for a classical value $v$ is the same, using the context $B[\blank] = \blank \parallel t \tags c?x.\ite{x = v}{t'\tags \tau. \nil}{\nil}$.
\qed\end{proof}

\begin{restatable}{theorem}{correct}\label{thm:correct}
	$\sim_l \ \subseteq \ \sim_{s}$.
\end{restatable}
\begin{proof}
	It is sufficient to show that $\sim_l$ is a saturated bisimulation up-to $Cv$ and contexts.
	Assume $\Delta \sim_{l} \Theta$.
	To see that $\mass{\Delta} = \mass{\Theta}$ it is sufficient to notice that $env(\Delta) = env(\Theta)$, by definition of $\sim_{l}$.
	Then we apply~\autoref{thm:envtrace}.

	We now show by induction over $\delta$ that for all $B[\blank], \Delta, \Theta, \Delta'$, if $\Delta \sim_{l} \Theta$ and $B[\Delta] \tauarrow_\delta \Delta'$ then
	$\Theta'$ exists such that $B[\Theta] \tauarrow_\delta \Theta'$ and $\Delta'\ Cv(B(\sim_{l}))\ \Theta'$.
	If $\delta = \singleton{h}$ then the property is trivially true as $B[\Delta]$ cannot move with $\tauarrow_{\overline{h}}$.

	Let $\delta = \singleton{t}$, then by~\autoref{thm:simplemoves} we have to consider three cases.
	If $\Delta' = B[\Delta'']$ for some $\Delta'' \neq \overline{\bot}$ such that $\Delta \tauarrow_{\singleton{t}} \Delta''$, then by hypothesis that $\Delta \sim_{l} \Theta$ it holds that $\Theta \tauarrow_{\singleton{t}} \Theta' \neq \overline{\bot}$ and then $B[\Theta] \tauarrow_{\singleton{t}} B[\Theta']\ B(\sim_{l})\ \Delta'$.
	If $\Delta' = \bot$ then $\Delta \tauarrow_{\singleton{t}} \overline{\bot}$ and by $\Delta \sim_{l} \Theta$ it holds that $\Theta \tauarrow_{\singleton{t}} \overline{\bot}$ and then $B[\Theta] \tauarrow_{\singleton{t}} \overline{\bot}$.
	
	The last case is $\Delta' = \sum_i p_i B_i'[\nE_i(\Delta)]$, with $\E_i$ trace non-increasing superoperators.
	Then since $env(\Delta) = env(\Theta)$, $B[\Theta] \tauarrow_t \Theta' = \sum_i p_i B_i'[\nE_i(\Theta)]$.
	Note that $\Delta \sim_{l} \Theta$ implies $\nE(\Delta) \sim_{l} \nE(\Theta)$, hence $\Delta'\ Cv(B(\sim_{l}))\ \Theta'$.
	

	Let $\delta = \singleton{(t_1,t_2)}$, then by~\autoref{thm:syncmoves} we have to consider four cases.
	The first three are similar to the ones for $\delta = \singleton{t}$.
	Take the last, and assume $\Delta' = B'[\Delta'']$ with $\Delta \xlongrightarrow{\mu}_{\singleton{t_1}} \Delta''$.
	Then, by hypothesis, $\Theta \xlongrightarrow{\mu}_{\singleton{t_1}} \Theta'$ with $\Delta'' \sim_{l} \Theta'$, and by~\autoref{thm:syncmoves}, $B[\Theta] \tauarrow_{\singleton{(t_1, t_2)}} B'[\Theta']\ B(\sim_{l})\ \Delta'$.


	Let $\delta = \delta_1 \psum{p} \delta_2$.
	Then, by definition of the scheduler $\delta_1 \psum{p} \delta_2$, $B[\Delta] \tauarrow_{\delta_1} \Delta_1$ and $B[\Delta] \tauarrow_{\delta_2} \Delta_2$ with $\Delta' = \Delta_1 \psum{p} \Delta_2$.
	By induction hypothesis, $B[\Theta] \tauarrow_{\delta_1} \Theta_1$ and $B[\Theta] \tauarrow_{\delta_2} \Theta_2$ with $\Delta_i\ Cv(B[\sim_{l}])\ \Theta_i$ for $i = 1,2$.
	Then, by linearity of $\tauarrow$ and definition of $\delta_1 \psum{p} \delta_2$, $B[\Theta] \tauarrow_{\delta_1 \psum{p} \delta_2} \Theta_1 \psum{p} \Theta_2\ \ Cv(Cv(B(\sim_{l})))\ \ \Delta_1 \psum{p} \Delta_2$.
	Finally, it is sufficient to see that $Cv(Cv(\rel)) = Cv(\rel)$ for any $\rel$.
\qed\end{proof}

\corrcompl*
\begin{proof}
	By~\autoref{thm:complete} and~\ref{thm:correct}.
\qed\end{proof}

\end{document}